\documentclass[11pt]{article}

\usepackage[a4paper, margin=1.2in]{geometry}
\usepackage{fourier} 
\usepackage{footmisc}
\usepackage{microtype} 
\usepackage{amssymb}
\usepackage{amsmath}
\usepackage{amsthm}
\usepackage{orcidlink}
\usepackage{academicons}
\usepackage{nicefrac}
\usepackage[utf8]{inputenc} 
\usepackage[T1]{fontenc}    
\usepackage{booktabs}       
\usepackage{wrapfig}
\usepackage{subcaption}
\usepackage{tikz}
\usetikzlibrary{arrows.meta}
\usepackage{algorithm}
\usepackage[noend]{algorithmic}

\usepackage[square]{natbib}

\usepackage[dvipsnames]{xcolor} 
\definecolor{orcidgreen}{HTML}{85A12C} 
\usepackage{url}

\usepackage{mathtools}
\newcommand{\defeq}{\coloneqq}
\newcommand{\eqdef}{\eqqcolon}

\newtheorem{theorem}{Theorem}[section]
\newtheorem{lemma}[theorem]{Lemma}

\newtheorem{corollary}[theorem]{Corollary}

\newtheorem{finding}[theorem]{Finding}

\usepackage{hyperref} 
\hypersetup{
    colorlinks=true,
    allcolors=orcidgreen,
    linkcolor=Blue,
    citecolor=purple,
}

\usepackage{cleveref}

\AddToHook{env/theorem/begin}{\crefalias{section}{theorem}}

\AddToHook{env/lemma/begin}{\crefalias{theorem}{lemma}}
\AddToHook{env/proposition/begin}{\crefalias{theorem}{proposition}}
\AddToHook{env/corollary/begin}{\crefalias{theorem}{corollary}}
\AddToHook{env/observation/begin}{\crefalias{theorem}{observation}}
\AddToHook{env/conjecture/begin}{\crefalias{theorem}{conjecture}}
\AddToHook{env/finding/begin}{\crefalias{theorem}{finding}}

\AddToHook{env/definition/begin}{\crefalias{section}{definition}}
\AddToHook{env/example/begin}{\crefalias{section}{example}}
\AddToHook{env/remark/begin}{\crefalias{section}{remark}}
\AddToHook{env/claim/begin}{\crefalias{section}{claim}}

\crefname{observation}{observation}{observations}
\Crefname{observation}{Observation}{Observations}
\crefname{conjecture}{conjecture}{conjectures}
\Crefname{conjecture}{Conjecture}{Conjectures}
\crefname{claim}{claim}{claims}
\Crefname{claim}{Claim}{Claims}
\crefname{finding}{finding}{findings}
\Crefname{finding}{Finding}{Findings}
\crefname{appendix}{appendix}{appendices}
\Crefname{appendix}{Appendix}{Appendices}

\usepackage{etoolbox}
\newif\ifproofsInAppendix
\proofsInAppendixtrue        

\newcommand{\appendixProofs}{}

\newcommand{\appendixproof}[3]{%
  \ifproofsInAppendix
    \gappto{\appendixProofs}{%
      \subsection{Proof of~{#1}~\ref{#2}}\label{proof:#2}%
      #3%
    }%
  \else
    \begin{proof}[Proof of {#1}~\ref{#2}]%
      #3%
    \end{proof}%
  \fi
}

\newcommand{\reals}{\mathbb{R}}
\newcommand{\pref}{\succ}

\newcommand{\out}{{\mathrm{out}\hbox{-}\mathrm{div}}}
\newcommand{\outdiv}{\out}
\newcommand{\ansd}{{\mathrm{ansd}}}
\newcommand{\swap}{{\mathrm{swap}}}

\newcommand{\calL}{{\mathcal{L}}}
\newcommand{\calD}{D}

\newcommand{\roott}{\rho}
\newcommand{\Var}{\mathrm{Var}}
\newcommand{\Cov}{\mathrm{Cov}}

\newcommand{\np}{{{\mathrm{NP}}}}

\newcommand{\SP}{{\mathrm{SP}}}
\newcommand{\SPG}{{\mathrm{SP/G}}}
\newcommand{\SPT}{{\mathrm{SP/T}}}

\newcommand{\SPstar}{{\mathrm{SP/star}}}
\newcommand{\SC}{{\mathrm{SC}}}
\newcommand{\GS}{{\mathrm{GS}}}
\newcommand{\GSbal}{{\mathrm{GS/bal}}}
\newcommand{\GScat}{{\mathrm{GS/cat}}}

\DeclareSymbolFont{bbold}{U}{bbold}{m}{n}
\DeclareSymbolFontAlphabet{\mathbbold}{bbold}

\usepackage{graphicx}
\graphicspath{{img/}}

\title{\huge \textbf{Outer Diversity of Condorcet Domains}}

\author{
\setlength{\tabcolsep}{15pt}  
\begin{tabular}{c c c}
  Piotr Faliszewski&
  Jan Jabrocki&
  Mateusz Słuszniak\\[2mm]
  Krzysztof Sornat&
  Stanisław Szufa&
  Tomasz Wąs
\end{tabular}
}

\date{}

\begin{document}

\maketitle

\begin{quote}
    \textbf{Abstract:}
  A Condorcet domain is a set of rankings over a given candidate set,
  such that every election that consists only of (an odd number of)
  votes from the domain has a transitive majority relation.  We study
  outer diversity of Condorcet domains, i.e., a measure that
  quantifies expected swap distance from a random vote to a closest
  one in the domain.  We numerically analyze outer diversity for
  maximal Condorcet domains with few candidates, and then we establish
  its asymptotic behavior for several special domains, mostly
  obtaining theoretical results.
\end{quote}
\begin{quote}
    \textbf{Code:} \url{https://github.com/Project-AGREE/Outer-Diversity-Condorcet/}
\end{quote}

\begingroup
\makeatletter
\renewcommand{\@makefnmark}{}
\makeatother
\footnotetext[0]{
  \hspace{-19pt}
  Authors' Information:
  \href{https://orcid.org/0000-0002-0332-4364}{Piotr Faliszewski \orcidlink{0000-0002-0332-4364}}, faliszew@agh.edu.pl, AGH University, Poland;
  {Jan Jabrocki}, jjabrocki@student.agh.edu.pl, AGH University, Poland;
  {Mateusz Słuszniak}, msluszniak1@gmail.com, AGH University, Poland; 
  \href{https://orcid.org/0000-0001-7450-4269}{Krzysztof~Sornat~\orcidlink{0000-0001-7450-4269}}, sornat@agh.edu.pl, AGH University, Poland;
  \href{https://orcid.org/0000-0001-6301-6227}{Stanisław Szufa \orcidlink{0000-0001-6301-6227}}, s.szufa@gmail.com, University of Geneva, Switzerland;
  \href{https://orcid.org/0000-0003-3492-6584}{Tomasz Wąs \orcidlink{0000-0003-3492-6584}}, tomasz.was@cs.ox.ac.uk, University of Oxford, United Kingdom.
}
\endgroup

\newcommand{\emphh}[1]{#1}

\section{Introduction}

We study outer diversity of Condorcet domains for ordinal elections,
focusing on both the values for specific domains and on algorithms for
computing them. Below we explain these notions intuitively and explain
why seeking outer diversity of Condorcet domains is interesting.

In the setting of ordinal preferences, each \emphh{vote} is a ranking
of candidates---from the most to the least appealing one---and a
\emphh{preference domain} is a set of votes that are allowed in a
considered election. For example, the \emphh{general domain} contains
all possible rankings, the \emphh{single-peaked} domain contains those
that arise when the candidates form a single-dimensional spectrum%
~\citep{bla:b:polsci:committees-elections},\footnote{E.g., this domains captures political preferences, when the candidates are ordered on the
left-to-right scale.} and the \emphh{group-separable} domain
captures preferences over candidates with hierarchically-organized
features~\citep{ina:j:group-separable,ina:j:simple-majority} (the
feature-based interpretation is presented, e.g., by
\citet{fal-kar-obr:c:group-separable}).
A \emphh{Condorcet domain} guarantees that the majority relation of an
election with an odd number of votes from such a domain
is transitive: If a majority of voters prefers candidate $a$ to candidate $b$,
and a (possibly different) majority of voters prefers $b$ to $c$, then there is also a majority of voters that prefers $a$ to $c$.
For example, single-peaked and group-separable domains are Condorcet
(see the book of \citet{pup-sli:b:condorcet-domains} and the surveys
of \citet{elk-lac-pet:t:restricted-domains-survey} and
\citet{mon:b:condorcet-survey} for a discussion of various such
domains).
Outer diversity, introduced by
\citet{fal-sor-szu-was:t:outer-diversity}, quantifies the expected
distance from a random vote to the closest one in the domain (see
\Cref{sec:prelim} for formal definitions) and, hence, measures how
well a given domain covers the general one.

Condorcet domains are among the most classic objects studied within
social choice theory. Indeed, the single-peaked domain was already
introduced by \citet{bla:b:polsci:committees-elections} and was
quickly generalized by \citet{arr:b:polsci:social-choice}.
\citet{sen:j:condorcet-never-conditions} provided an exact
characterization of Condorcet domains using so-called \emphh{never
conditions}: A domain is Condorcet if and only if for each triple of
candidates $a$, $b$, and $c$ it either satisfies a never-top, a
never-middle, or a never-bottom condition: The former requires that
there is candidate $x \in \{a,b,c\}$ such that no vote in the domain
ranks $x$ on top when restricted to $\{a,b,c\}$,
and the other two conditions are defined analogously.
For example, Black's and Arrow's single-peaked domains satisfy
never-bottom conditions, peak-pit domains
of~\citet{dan-kar-kos:j:condorcet-peak-pit} satisfy combinations of
never-top and never-bottom conditions (and, hence, peak-pit domains
include the single-peaked ones, but
also many others). On the other hand, group-separable domains satisfy
never-middle conditions only.
Generally, the never conditions satisfied by a given domain have a
strong influence on its properties, but the nature of this influence
is only partially understood.  We are interested in outer diversity of
Condorcet domains because it extends this understanding.

Diversity analysis is also interesting because it gives a new
perspective on the long-studied problem of finding large Condorcet
domains~\citep{fis:j:condorcet-fishburn,gal-rei:j:condorcet-big,dan-kar-kos:j:condorcet-peak-pit,kor-sli:j:condorcet-big,pup-sli:j:condorcet-max}.
Such domains are sought after because they are interesting
mathematically, but also because, intuitively, the larger a
domain is, the more freedom it gives the voters to express their views
(while ensuring transitivity of the majority relation).  Yet, as shown
by \citet{fal-sor-szu-was:t:outer-diversity}, having many votes does
not need to translate to high outer diversity value: The largest
Condorcet domain with $8$ candidates, computed by
\citet{lee-mar-rii:j:largest-con-domain8-condorcet8}, has 224 votes
but is notably less diverse than some group-separable ones, with only
128 votes.
As outer diversity quantifies voters' ability to express their views
more directly than domain cardinality, we propose to seek Condorcet
domains for which outer diversity is highest.

Further, diversity analysis of Condorcet domains is currently
receiving increased attention.  For example,
\citet{kar-mar-rii-zho:t:domain-diversity} and
\citet{amm-pup:j:domain-diversity} measured diversity using various
richness properties, such as the number of different tuples of
candidates that can be ranked among top positions,
and \citet{fal-sor-szu-was:c:structured-kemeny} introduced
\emphh{inner diversity}, which measures how tightly packed are the
votes in a given domain.  The advantage of outer diversity over these
measures is that it provides a single, easy to interpret number (inner
diversity provides vectors of numbers, which are not easy to compare,
and it is not clear which richness-based notions to use).  For a
general theory of diversity, we point to the work of
\citet{neh-pup:j:diversity}.

\paragraph{Contributions.}
First, we consider all Condorcet domains with six candidates, as
enumerated by \citet{ake-lee-lit-mar-rii:j:condorcet7}, and we analyze
their properties, connecting these properties to the never conditions
that they satisfy.
Based on these findings,
we pursue asymptotic behavior of outer diversity of several special
domains.
While doing so, we improve some of the results of
\citet{fal-sor-szu-was:t:outer-diversity}. For example, they computed
outer diversity of their domains for up to $20$ candidates, whereas we
consider up to $100$ candidates (albeit we do not consider all the
domains that they do), they showed a lower bound on outer diversity of
the group-separable/caterpillar domain and we give an exact formula
for its value, and they provided approximation, sampling-based
algorithms for computing outer diversity for their domains, whereas we
give some exact algorithms.

Additionally, we also consider domains \emphh{single-peaked on a
  tree}~\citep{dem:j:sp-trees,pet-yu-cha-elk:j:sp-on-tree}, which are
\emphh{weak Condorcet}.
Any election with an odd
number of votes from such a domain is guaranteed to have a Condorcet
winner, i.e., a candidate that is preferred to every other one by a
majority of voters. We find that, in the limit, such domains can be as
outer-diverse as the general one.
All omitted proofs can be found in \Cref{app:missing-proofs}.

\section{Preliminaries}\label{sec:prelim}

Given a set of candidates $C = \{c_1, \ldots, c_m\}$, a \emph{vote}
(or, a \emph{preference order}, or a \emph{preference ranking}) is a
total order over $C$, ranking the candidates from the most to the
least desirable one. For two candidates $a, b \in C$, we write
$v \colon a \pref b$ to indicate that  $a$ is preferred over $b$ in vote $v$ (i.e., $v$ ranks $a$ higher than $b$). 
An election $E = (C,V)$ is a pair, where $C$ is a set of candidates
and $V$ is a collection of voters, each with a single vote. For an
election $E = (C,V)$, we say that candidate $a \in C$ is
\emph{majority-preferred} to candidate $b \in C$, denoted
$a \mathrel{M_E} b$, if a strict majority of voters in $V$ prefer $a$
over $b$. We refer to $M_E$ as the \emph{majority relation} of $E$.
We are particularly interested in elections for which $M_E$ is an
order, and we call such elections \emph{transitive}.  A candidate
$c \in C$ is a Condorcet winner if for every other candidate $d \in C$
we have $c \mathrel{M_E} d$. Naturally, an election either has a unique
Condorcet winner or does not have one at all. Every transitive
election has a Condorcet winner, but there are also nontransitive ones
with Condorcet winners.

\subsection{Domains}\label{sec:domains}

A \emph{domain} over candidate set $C$ is any subset of votes over~$C$.
In particular, we write $\calL(C)$ to denote the \emph{general
  domain} that contains all votes over $C$.  By a domain's size, we mean
  the number of votes that it contains. Given a domain
$\calD$, we say that an election $E$ is a $\calD$-election if every
voter has a vote that belongs to $\calD$.  
We say that $\calD$ is a \emph{Condorcet domain} if every
$\calD$-election with an odd number of voters is transitive.  A
Condorcet domain $\calD$ is \emph{maximal} if for every vote $v$ (over
the same candidate set) not present in $\calD$, $\calD \cup \{v\}$ is
not a Condorcet domain.  $\calD$ is a \emph{weak Condorcet domain} if
every $\calD$-election with an odd number of voters has a Condorcet
winner. For a detailed survey on Condorcet domains, we point to the
book of \citet{pup-sli:b:condorcet-domains}, as well as to the surveys
of \citet{elk-lac-pet:t:restricted-domains-survey} and
\citet{mon:b:condorcet-survey}.

Condorcet domains are characterized using so-called \emph{never
  conditions}.  For a triple $\{a, b, c\} \subseteq C$ we say that $x \in \{a,b,c\}$
is \emph{never-top} (resp. \emph{never-middle}, \emph{never-bottom})
in $\calD$ with respect to $\{a, b, c\}$ if no vote in $\calD$ ranks
$x$ first (resp. second, last) among $\{a, b, c\}$.  A domain is a
Condorcet domain if and only if, for every triple of candidates, it
satisfies at least one of the
never conditions~\citep{sen:j:condorcet-never-conditions}.
Next we describe the domains and families of domains that we focus on.

\begin{description}
\item[Peak-Pit.] A domain $\calD$ is a peak-pit domain if it is a
  Condorcet domain that for every triple of candidates either
  satisfies a never-top or a never-bottom
  condition~\citep{dan-kar-kos:j:condorcet-peak-pit}.
\end{description}
While peak-pit domains are one of the major classes of Condorcet
domains, so far they only received very limited attention from
computational social choice.
For example, the complexity of deciding if a given set of never-top
and never-bottom conditions yields a nonempty domain is not known
(\citet{pup-sli:b:condorcet-domains} conjecture that it is
$\np$-complete).

A Fishburn domain is a peak-pit domain that for a given ordering of
candidates, $c_1 \rhd c_2 \rhd \cdots \rhd c_m$, requires that among
every three candidates $c_i,c_j,c_k$, with $i < j < k$, $c_j$ is
never-top, if $j$ is odd, and $c_j$ is never-bottom, if $j$ is
even~\citep{fis:j:condorcet-fishburn}.  Fishburn domains are
among the largest Condorcet domains
(albeit, in general, they are known to not be the largest ones~\citep{lee-mar-rii:j:largest-con-domain8-condorcet8}).

As opposed to peak-pit domains, the family of single-peaked domains
has received very detailed attention in (computational) social choice and
includes some of the best known
domains~\citep{elk-lac-pet:t:restricted-domains-survey}. However, not
all domains in this family are Condorcet or even weak Condorcet.
Formally, let~$G$ be an undirected, connected graph with vertex set
$C$.  Following \citet{fal-sor-szu-was:t:outer-diversity}, we refer to
such graphs as SP-graphs (or SP-trees, when $G$ is a tree). A vote
$v \in \calL(C)$ is single-peaked on $G$ if for each $t \in [m]$ the
$t$ top-ranked candidates in $v$ form a connected component within $G$.
A domain $\calD$ is \emph{single-peaked on a graph}, denoted $\SPG$, if
there is an SP-graph $G$ such that $D$ consists exactly of votes
single-peaked on $G$. The particular $G$ we have in mind will either
be clear from the context or irrelevant for the discussion. We
consider the following $\SPG$ domains:
\begin{description}
\item[Single Peaked (SP).] The classic \emph{single-peaked domain
    ($\SP$)}, introduced by \citet{bla:b:polsci:committees-elections},
  is defined over a path and, w.l.o.g., we use path
  $c_1 \rhd c_2 \rhd \cdots \rhd c_m$. This path is called the
  \emph{societal axis}.  SP is a Condorcet domain: Indeed, for every
  three candidates $c_i, c_j, c_k$, with $i < j < k$, we see that
  $c_j$ is never-bottom. $\SP$ is sometimes referred to as
  \emph{Black's single-peakedness}.  We are also interested in Arrow's
  single peaked domains (Arrow's SP domains, for short) that
  generalize Black's SP by requiring that every triple of candidates
  satisfies some never-bottom condition, but not restricting these
  conditions to be defined by the
  axis~\citep{arr:b:polsci:social-choice}.\footnote{The reader may
    wonder about domains defined by never-top conditions only. These
    are so-called \emph{single-dipped} domains and are obtained by
    reversing all the votes in the Arrow's single peaked ones. Because
    of this symmetry, they are not interesting in our study as,
    by necessity, their outer diversity is the same as that of their
    dual Arrow's $\SP$ domains.}
  Note that Arrow's SP domains are not $\SPG$ domains.
\item[SP/T.]
  If an SP-graph is a tree, then we refer to it as an $\SP$-tree and
  to the domain as \emph{single-peaked on a tree ($\SPT$)}. These
  domains were introduced by \citet{dem:j:sp-trees}.  A specific
  $\SPT$ that we will consider is based on a star graph with $m$
  leaves---we call such domain $\SPstar_m$.  $\SPT$ are weak Condorcet
  domains, but not Condorcet ones. Algorithmic study of $\SPT$ domains
  was provided by \citet{pet-yu-cha-elk:j:sp-on-tree}.
\end{description}
If $G$ is an SP-graph with a cycle, then the $\SPG$ domain that
it yields is not even weak Condorcet.
The family of single-crossing domains was introduced by
\citet{mir:j:single-crossing} and \citet{rob:j:tax}:
\begin{description}
\item[Single Crossing (SC).] A domain $\calD$ is
  \emph{single-crossing} if it is possible to order its votes as
  $v_1 \rhd v_2 \rhd \cdots \rhd v_n$, so that for each two candidates
  $c_i, c_j \in C$, votes where $c_i$ is preferred to $c_j$ form
  either a prefix or a suffix of this order. SC domains are Condorcet
  and, indeed, they form a subclass of the peak-pit domains.
\end{description}

We note that there are many nonisomorphic single-crossing
domains~\citep{fal-sko-sli-szu-tal:j:isomorphism}. In particular,
\citet{szu-boe-bre-fal-nie-sko-sli-tal:j:map} suggested an algorithm
for sampling such domains, albeit their algorithm does not ensure
uniform distribution.  Whenever we refer to single-crossing domains,
either it will be clear which specific $\SC$ domain we mean, or it
will not be relevant.

Next we consider group-separable domains. While these domains were
introduced by \citet{ina:j:group-separable,ina:j:simple-majority}, we
use their definition based on trees, provided by
\citet{kar:j:group-separable}.  Let $T$ be a rooted, ordered, binary
tree
where each leaf is labeled with a unique candidate from $C$. We refer to
such trees as \emph{GS-trees}. A \emph{frontier} of a GS-tree $T$ is the ranking
of candidates obtained by reading them off from the leaves, from left
to right. We say that a vote $v \in \calL(C)$ is compatible with $T$
if $v$ is a frontier of some tree obtained from $T$ by, possibly,
reversing the order of some of its nodes' children:
\begin{description}
\item[Group-Separable (GS).] A domain $\calD$ is
  \emph{group-separable} ($\GS$) if there is a GS-tree $T$ such that
  $\calD$ consists exactly of the votes compatible with $T$. $\GS$
  domains are Condorcet and satisfy never-middle condition, but not
  all domains defined by never-middle conditions are $\GS$.
\item[GS/cat and GS/bal.] A GS-tree is a \emph{caterpillar} if every
  internal node has at most one nonleaf child. A GS/caterpillar domain
  ($\GScat$) is a $\GS$ domain defined by a caterpillar
  tree. Similarly, a GS-tree is balanced if for each two leaves their
  distance from the root differs at most by one, and GS/balanced
  ($\GSbal$) is a $\GS$ domain defined by a balanced tree. Note that
  if we ignore leaf labels, then for a given number of candidates,
  there is a unique caterpillar tree, but possibly many balanced ones,
  leading to nonisomorphic domains.
\end{description}
It is well-known that for a size-$m$ candidate set, Black's $\SP$, as
well as all $\GS$ domains, contain exactly $2^{m-1}$ votes. The same
holds for Arrow's $\SP$ domains~\citep{sli:j:arrow-sp}.

\subsection{Outer Diversity}\label{sec:outdiv}
Given two votes $u, v$ over candidate set $C$, their \emph{swap distance},
known also as \emph{Kendall tau distance},
denoted $\swap(u,v)$, is 
the smallest number of swaps of adjacent candidates needed to transform 
one into the other.
The \emph{average normalized swap distance} of a domain $\calD$ is defined~as:
\[ \ansd(\calD) = \frac{\sum_{u \in
    \mathcal{L}(C)} \swap(\calD,u)}{m! \cdot \binom{m}{2}},\]
i.e., the
expected swap distance between a vote selected uniformly at random
from $\mathcal{L}(C)$ and its nearest member of $\calD$, normalized by
the maximum possible swap distance. The \emph{outer diversity} of a
domain $D \subseteq \mathcal{L}(C)$ is
$\out(\calD) = 1 - 2\cdot\ansd(\calD)$.  Outer diversity assumes
values between $0$ and $1$, where $0$ is the lowest possible
diversity, achieved, e.g., by each singleton domain $\{v\}$, and $1$
is the maximum diversity, achieved by $\calL(C)$.  Outer diversity was
introduced by \citet{fal-sor-szu-was:t:outer-diversity}.  In
particular, they gave a general brute-force algorithm for computing
its value for domains specified by listing their votes (running in
time $O(m\cdot m!)$, where $m$ is the number of candidates; we use
this algorithm as a default), as well as sampling-based
algorithms. For the latter, the idea is to sample votes from
$\calL(C)$ and compute their closest members from the considered
domain: They provided polynomial-time algorithms for computing a
closest $\SP$, $\SC$, and $\GS$ vote for a given one (for $\SP$, see
also the work of \citet{fal-hem-hem:j:nearly-sp}), and a
non-polynomial algorithm for $\SPT$.  They computed outer diversity
values for a few special domains for up to $20$ candidates, gave lower
bound of $\nicefrac{1}{2}$ for $\GScat$, and obtained a few other
results less directly relevant to our work.

\section{Outer Diversity of Domains With a Few Candidates}
\label{sec:few}

We begin our analysis with an extensive study of Condorcet domains
with a few candidates. In this case it is possible to enumerate all
nonisomorphic maximal Condorcet domains and we use such an enumeration
for $6$ candidates, provided by
\citet{ake-lee-lit-mar-rii:j:condorcet7} (they generated domains with
up to 7 candidates, and we show results for a sample of these
domains in \Cref{app:seven-candidates}; the overall conclusions are
the same.)
Additionally, we also discuss single-crossing domains (which typically
are not maximal Condorcet domains) and $\SPT$ domains.\footnote{We ran
  the experiments presented in the main body of the paper on personal
  computers. The computatoin required a few dozen hours.}

\subsection{Maximal Condorcet Domains}
As established by \citet{ake-lee-lit-mar-rii:j:condorcet7}, for the
case of $6$ candidates there are $256\,895$ nonisomorphic maximal
Condorcet domains. Among them there are $9\,939$ peak-pit domains, of
which $40$ are Arrow's SP domains---including the classic Black's
single-peaked domain---and another $40$ are single-dipped
domains. There are also $12$ domains defined by never-middle
conditions only, of which~$6$ are group-separable (two of which have
the same outer diversity).  In~\Cref{fig:all_condorcet} we plot all
these domains as dots, where the $x$ coordinate gives the number of
votes in the domain and the $y$ coordinate gives its outer diversity.
Different families of domains are depicted with different colors. As
there are $256\,895$ domains, many points are occluded by others. For
example, many peak-pit domains are drawn ``on top of'' other domains.
We proceed by analyzing the results for particular families of
domains.

\begin{figure}[t]
    \centering
    \begin{subfigure}[t]{0.485\textwidth}
        \centering
        \includegraphics[height=5cm]{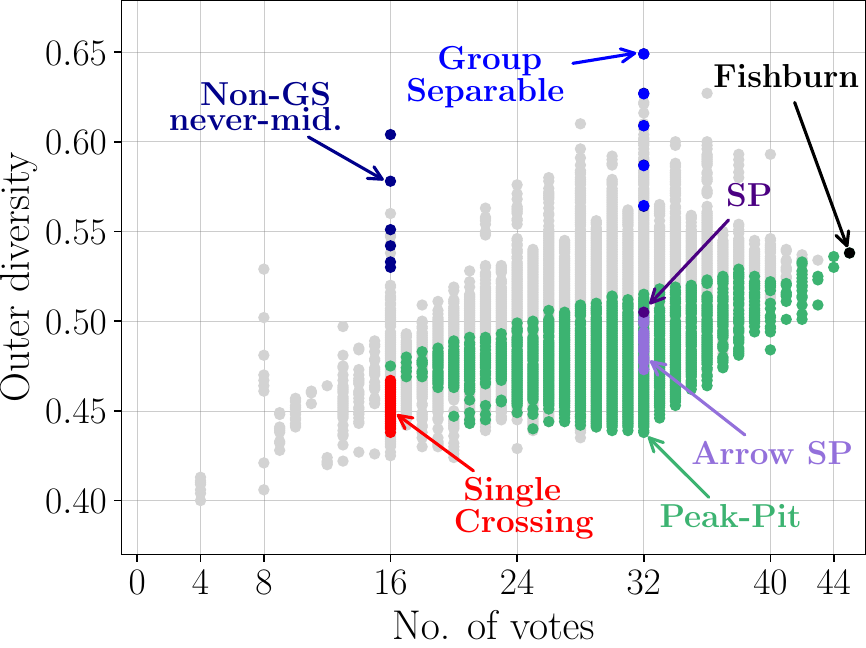}
        \caption{\label{fig:all_condorcet}Every maximal nonisomorphic
          Condorcet domain with six candidates is placed as a dot, with
          its $x$ coordinate giving the domain's size,
          and its $y$ coordinate giving its outer diversity. If a
          domain belongs to a certain class of domains, such as peak-pit,
          then its dot is colored accordingly; otherwise it is
          gray. We additionally include a sample of 200
          single-crossing domains.}
    \end{subfigure}
    \hfill
    \begin{subfigure}[t]{0.485\textwidth}
        \centering
        \includegraphics[height=5cm]{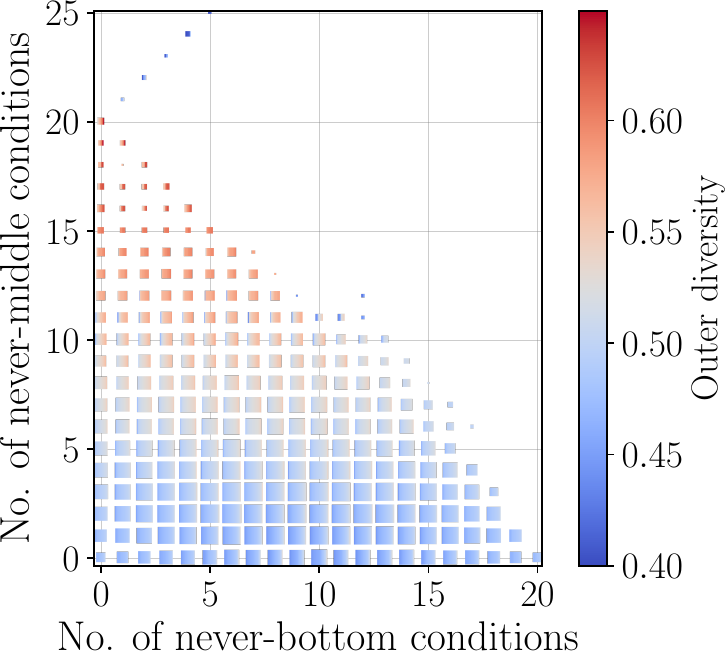}
        \caption{\label{fig:botmid}Outer diversity of maximal
          nonisomorphic Condorcet domains depending on the number of
          satisfied never-bottom ($x$-axis) and never-middle
          ($y$-axis) conditions. The size of each square reflects the
          number of domains satisfying given numbers of conditions, and
          the horizontal color gradient presents the
          distribution of outer diversity values for these domains.}
    \end{subfigure}
    \caption{Analysis of all of the maximal Condorcet domains for six candidates.}
    \label{fig:fullcond6}
\end{figure}
  
We note that maximal Condorcet domains with even numbers of votes tend to
achieve more varied and larger outer diversity values than those with
odd numbers.
Next, for a given number of votes, peak-pit domains tend to assume lower
outer diversity values than domains that include never-middle
conditions. Indeed, domains defined entirely by never-middle
conditions have some of the highest outer-diversity values overall
(including the highest one, assumed by $\GScat$).  We explore this in
\Cref{fig:botmid}. Specifically, for each pair $b, m$ of integers, we
let $\calD_6(b,m)$ be the subset of our domains that satisfy exactly
$b$ never-bottom and $m$ never-middle conditions (they also typically
satisfy $20-b-m$ never-top conditions, but a small number of domains
satisfies more than $20$ never conditions in total). For each
$\calD_6(b,m)$, we plot a square centered over point $(b,m)$, whose
side's length is proportional to $1+\log|\calD_6(b,m)|$. The color of
the square gradually shifts from left to right, proportionally to the
outer diversities of the domains in $\calD_6(b,m)$, sorted in the
non-decreasing order. Altogether, it is visible that the number of
satisfied never-middle conditions is strongly correlated with the
domains' outer-diversity (with PCC = $0.714$) and negatively correlated
with the number of domains (with PCC = $-0.565$).
One possible explanation of lower diversity of peak-pit domains is
that---at least for up to $7$ candidates---these are exactly the
connected domains, i.e., domains where for each two of their votes $u$
and $v$, one can transform $u$ into $v$ by swapping adjacent
candidates only, without ever leaving the
domain~\citep{ake-lee-lit-mar-rii:j:condorcet7}. This means that, as
compared to other domains of the same size, they are confined to a
smaller area of the space. On the other hand,
\citet{fal-sor-szu-was:t:outer-diversity} have shown that $\GScat$ is
very uniformly spread out through the space of all possible votes and,
in this sense, is strongly disconnected.

\begin{finding}
  The outer diversity of maximal nonisomorphic Condorcet domains with
  six candidates is strongly correlated with the number of satisfied
  never-middle conditions (PCC = $0.714$).
\end{finding}

Interestingly, for many pairs $(b,m)$ we observe outlier domains in
$\calD_6(b,m)$ with very low outer diversity values (see, e.g.,
vertical blue stripe on the left part of squares for $m \in \{10,11\}$
in \Cref{fig:botmid}). Typically these domains consist of very few
votes. However, while low cardinality must imply low outer diversity
(see the work of \citet{fal-sor-szu-was:t:outer-diversity} and the
next section), generally the correlation between cardinality and outer
diversity of a domain is fairly weak (PCC=$0.345$). For example, the
Fishburn domain, which is the largest maximal Condorcet domains for
six candidates (with 45 votes), has outer diversity of $0.538$ (the
average over all maximal Condorcet domains for six candidates is
$0.495$, with standard deviation $0.021$, and $\GScat$ has
outer diversity $0.649$).

Let us now consider single-peaked domains.  For the case of $6$
candidates, Black's SP domain is the most diverse Arrow's SP domain
(the same holds for $m=7$). It would be very interesting to prove this
formally, or to verify computationally for larger numbers of
candidates, but currently it is not clear how to efficiently generate
Arrow's SP domains.
\begin{finding}
  Black's SP has the highest outer diversity among Arrow's SP domains
  for 6 candidates.
\end{finding}

\newenvironment{gstree}{
  \begin{tikzpicture}[
    trim left,
    node_style/.style={circle, fill=black, inner sep=1pt}, 
    edge_style/.style={thin},
    scale = 0.25]\path (0,0);}{\end{tikzpicture}
}

\newcommand{\gstrees}{%
  \begin{tabular}{l}
    \begin{gstree}
        \node[node_style] (A) at (1,0) {};
        \node[node_style] (B) at (2,0) {};
        \node[node_style] (C) at (3,0) {};
        \node[node_style] (D) at (1.2,0.8) {};
        \node[node_style] (E) at (2,1) {};
        \node[node_style] (F) at (2.8,0.8) {};
        \draw[edge_style] (A) -- (B);
        \draw[edge_style] (B) -- (C);
        \draw[edge_style] (B) -- (D);
        \draw[edge_style] (B) -- (E);
        \draw[edge_style] (B) -- (F);
        \node at (9,0) {\small $240$};
        \node at (14,0) {\small $0.778$};
        \node at (9,2) {\small $n$};
        \node at (14,2) {\small $\mathrm{\outdiv}$};
      \end{gstree}
    \\
    \begin{gstree}
        \node[node_style] (A) at (1-0.5,0) {};
        \node[node_style] (B) at (2-0.5,0) {};
        \node[node_style] (C) at (3-0.5,0) {};
        \node[node_style] (D) at (4-0.5,0) {};
        \node[node_style] (E) at (3.8-0.5,-0.8) {};
        \node[node_style] (F) at (3.8-0.5,0.8) {};
        \draw[edge_style] (A) -- (B);
        \draw[edge_style] (B) -- (C);
        \draw[edge_style] (C) -- (D);
        \draw[edge_style] (C) -- (E);
        \draw[edge_style] (C) -- (F);
        \node at (9,0) {\small $132$};
        \node at (14,0) {\small $0.671$};
    \end{gstree}       
    \\
    \begin{gstree}
        \node[node_style] (A) at (0.5,-0.5) {};
        \node[node_style] (B) at (1.5,0) {};
        \node[node_style] (C) at (2.5,0) {};
        \node[node_style] (D) at (3.5,0.5) {};
        \node[node_style] (E) at (0.5,0.5) {};
        \node[node_style] (F) at (3.5,-0.5) {};
        \draw[edge_style] (A) -- (B);
        \draw[edge_style] (B) -- (C);
        \draw[edge_style] (C) -- (D);
        \draw[edge_style] (B) -- (E);
        \draw[edge_style] (C) -- (F);
        \node at (9,0) {\small $112$};
        \node at (14,0) {\small $0.644$};
    \end{gstree}\\
    \begin{gstree}
        \node[node_style] (A) at (0,0) {};
        \node[node_style] (B) at (1,0) {};
        \node[node_style] (C) at (2,0) {};
        \node[node_style] (D) at (3,0) {};
        \node[node_style] (E) at (4,0) {};
        \node[node_style] (F) at (2,1) {};
        \draw[edge_style] (A) -- (B);
        \draw[edge_style] (B) -- (C);
        \draw[edge_style] (C) -- (D);
        \draw[edge_style] (D) -- (E);
        \draw[edge_style] (C) -- (F);
        \node at (9,0) {\small \phantom{$0$}$72$};
        \node at (14,0) {\small $0.582$};
      \end{gstree}       
    \\
    \begin{gstree}
        \node[node_style] (A) at (0,0) {};
        \node[node_style] (B) at (1,0) {};
        \node[node_style] (C) at (2,0) {};
        \node[node_style] (D) at (3,0) {};
        \node[node_style] (E) at (4,0.5) {};
        \node[node_style] (F) at (4,-0.5) {};
        \draw[edge_style] (A) -- (B);
        \draw[edge_style] (B) -- (C);
        \draw[edge_style] (C) -- (D);
        \draw[edge_style] (D) -- (E);
        \draw[edge_style] (D) -- (F);
        \node at (9,0) {\small \phantom{$0$}$60$};
        \node at (14,0) {\small $0.573$};
    \end{gstree}\\
    \begin{gstree}
        \node[node_style] (A) at (0-0.5,0) {};
        \node[node_style] (B) at (1-0.5,0) {};
        \node[node_style] (C) at (2-0.5,0) {};
        \node[node_style] (D) at (3-0.5,0) {};
        \node[node_style] (E) at (4-0.5,0) {};
        \node[node_style] (F) at (5-0.5,0) {};
        \draw[edge_style] (A) -- (B);
        \draw[edge_style] (B) -- (C);
        \draw[edge_style] (C) -- (D);
        \draw[edge_style] (D) -- (E);
        \draw[edge_style] (E) -- (F);
        \node at (9,0) {\small \phantom{$0$}$32$};
        \node at (14,0) {\small $0.505$};
    \end{gstree}       
    \end{tabular}
}

\begin{figure}[t]
    \centering
    \begin{subfigure}[b]{0.32\linewidth}
      \centering
      \begin{tikzpicture}
        \node[anchor=south west, inner sep=0] (img2)        
        {\includegraphics[width=\linewidth]{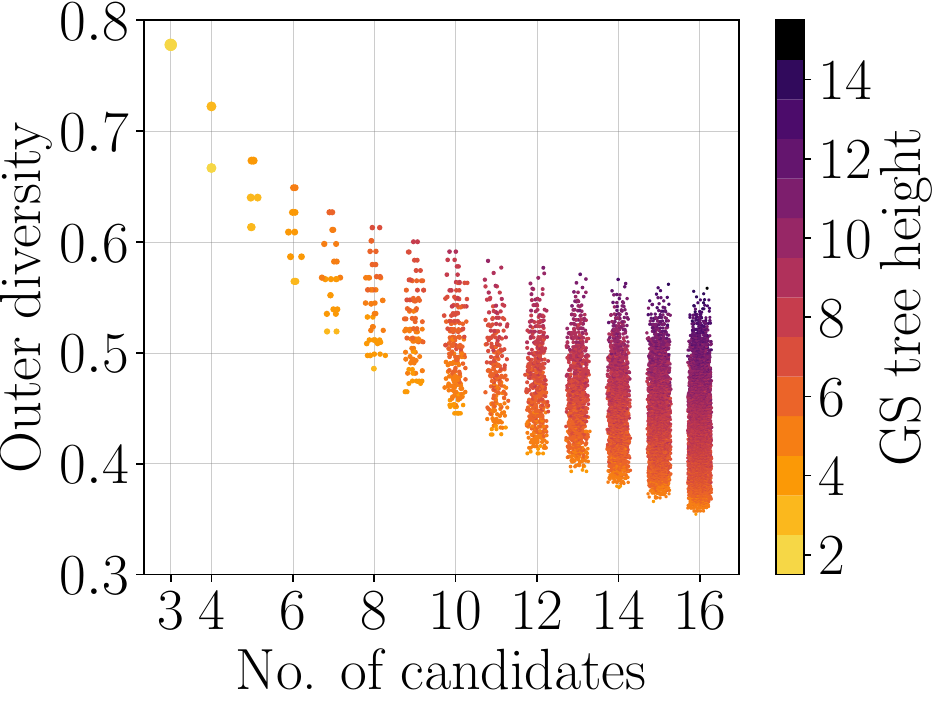}};
      \end{tikzpicture}
      \caption{Outer diversity of $\GS$ domains. Colors give the
        heights of the trees that generate the domains. $\GScat$ and
        $\GSbal$ have highest and lowest outer diversity,
        respectively.}
      \label{fig:gs-cands}
    \end{subfigure}
          \hfill
    \begin{subfigure}[b]{0.32\linewidth}
      \centering
      {\includegraphics[width=\linewidth]{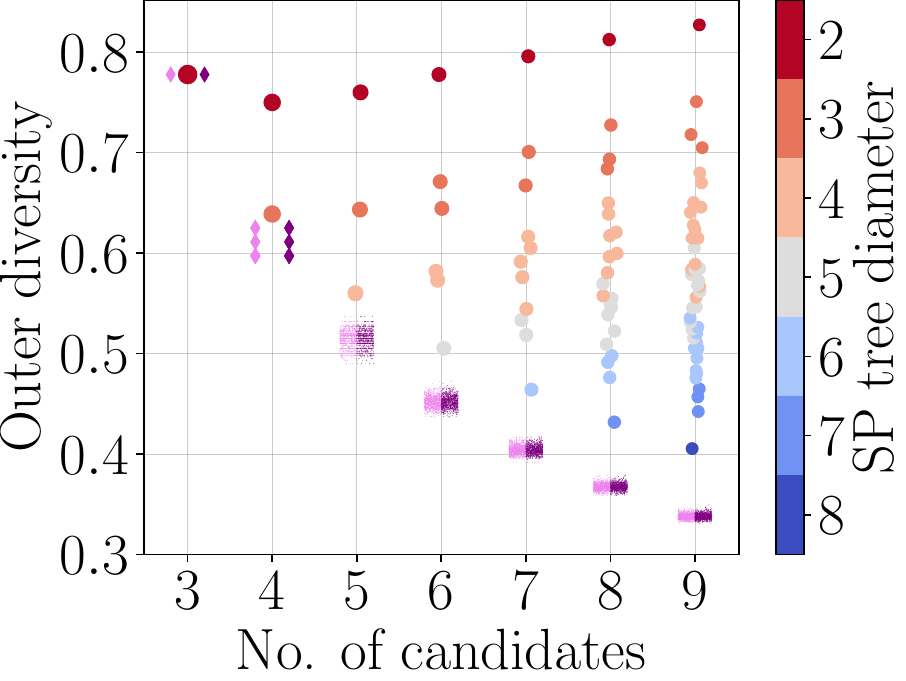}}
      \caption{Outer diversity of $\SC$ domains (light and dark purple dots; see text) and $\SPT$ domains (blue-to-red dots; colors
        give diameters of the trees generating the domains).}
      \label{fig:sc-spt-cands}
    \end{subfigure}
    \hfill
    \begin{subfigure}[b]{0.32\linewidth}
      \centering      
      \fbox{\scalebox{0.8}{\centering\quad\gstrees}}
      \caption{Nonisomorphic trees for $\SPT$ domains with $6$
        candidates. For each tree we  provide the number of votes
        ($n$) in the domain it generates and its outer diversity ($\outdiv$).}
      \label{fig:spt6}
    \end{subfigure}
    \caption{Outer diversity of $\GS$, $\SPT$ and $\SC$ domains, depending on the number of candidates.}
    \label{fig:three_images}
      \begin{tikzpicture}[overlay,yshift=3.22cm]
        \draw (-0.24,4) -- (0.1,4) -- (0.1,2.4) -- (-0.24,2.4) -- (-0.24,4);
        \draw[-{Latex[scale=1., fill=black]}] (0.1,3.86) -- (2.9,3.86);
      \end{tikzpicture}
  \end{figure}

  Next, let us consider group-separable domains.  As they achieve some
  of the largest outer diversity values among maximal Condorcet
  domains with six candidates, in \Cref{fig:gs-cands} we also explore
  how their outer diversity depends on the number of candidates and
  the heights of the trees used to generate them. Specifically, for
  each $m \in \{3, \ldots, 16\}$ we generated all nonisomorphic
  GS-trees with $m$ nodes and computed the outer diversities of their
  corresponding domains.  In \Cref{fig:gs-cands} we depict them as
  orange/purple dots (their $x$ coordinates correspond to the number of
  candidates, perturbed a bit for clarity, and their $y$ coordinates
  give their outer diversity values).  We see that the taller a given
  tree, the larger is the outer diversity value of a given domain.

\begin{finding}\label{fin:gs}
  For each number of candidates in $\{3, \ldots, 16\}$, $\GScat$ has
  the highest outer diversity and $\GSbal$ has the lowest one among
  $\GS$ domains. There is a strong correlation between the height of
  the domain's tree and its outer diversity (for each
  $m \in \{4, \ldots, 16\}$, PCC $\geq 0.896$; for $m=3$ there is only
  one domain).
\end{finding}

It is also interesting to consider the domains that satisfy
never-middle conditions only, but that are not group-separable. As we see in
\Cref{fig:all_condorcet}, for six candidates they have fewer votes than
the GS ones and their outer diversity values are a bit lower. In
\Cref{app:microscope} we have a closer look at them and find that they
are structurally quite similar to the GS ones.

\subsection{Single-Crossing Domains}

A single-crossing domain is maximal if it is not possible to extend it
without making it non-single-crossing.\footnote{Note that being
  a \emph{maximal single-crossing domain} typically does not mean
  being a maximal Condorcet domain.}
It is well-known that each
maximal single-crossing domain consists of $1+\binom{m}{2}$
votes,
which can be arranged in a line, where each two consecutive votes
differ by ordering of two adjacent candidates.
For this size, single-crossing domains have fairly low outer
diversity. However, as opposed to many other classes of peak-pit
domains, we can sample single-crossing domains uniformly at random
and, hence, evaluate how their outer diversity changes for somewhat
larger candidate sets. Unfortunately, our sampling algorithm runs in
exponential time.

\begin{theorem}\label{thm:sc-uniform}
  There exists a sampling procedure that, after $O(m \cdot m!)$ preprocessing,
  samples a single-crossing domain over $m$ candidates uniformly at random in $O(m^3)$ time.
\end{theorem}
\appendixproof{Theorem}{thm:sc-uniform}{ Consider candidate set
  $C = \{c_1, \ldots, c_{m}\}$.  First, we construct a swap graph,
  where each node is associated with a vote: Starting from vote
  $u \colon c_1 \pref c_2 \pref \cdots \pref c_m$, we generate all
  votes that can be obtained by swapping pairs of adjacent candidates
  $c_j, c_{j+1}$, where $c_j$ is ranked before $c_{j+1}$.  This
  produces a directed acyclic graph (DAG) whose nodes are votes and
  whose edges correspond to such swaps (directed toward the vote with
  larger swap distance from $u$).  This DAG has exactly $m!$ nodes and
  $O(m \cdot m!)$ edges.  A single-crossing domain corresponds to a
  maximal chain in this DAG.

  For each node $v$ in this DAG, we compute the number of maximal
  chains starting at $v$, denoted $f(v)$.  If $v$ is the (unique)
  sink, then $f(v) = 1$.  Otherwise, $f(v) = \sum_{w} f(w)$, where the
  sum ranges over the direct successors of $v$.  Since each node has
  at most $m - 1$ successors and there are $m!$ nodes, this phase runs
  in $O(m \cdot m!)$ time.

  Finally, to sample a domain, we start from the top node, vote $u$,
  and at each step select a direct successor $w$ with probability
  $f(w) / f(v)$, then continue from $w$ until the sink is reached.
  The resulting path has length $\binom{m}{2}$, and each step
  considers at most $m - 1$ successors, so sampling runs in $O(m^3)$
  time.  Sampling proportionally to the counts ensures that every
  maximal chain, i.e., single-crossing domain, is reached with equal
  probability.
}

Using this sampler, for each $m \in \{3, \dots, 9\}$ we sampled
$1000$ single-crossing domains uniformly at random, without
repetition. Additionally, we sampled $1000$ such domains using the
polynomial-time algorithm of
\citet{szu-boe-bre-fal-nie-sko-sli-tal:j:map} (which does not achieve
uniform distribution). For $m \in \{3,4,5\}$ there are fewer than
$1000$ maximal single-crossing domains and for this case we generated
all such domains.
We restricted attention to up to $9$ candidates due to speed of exact
computation of outer diversity.
For each of the domains we computed its outer diversity.  We show the
results in~\Cref{fig:sc-spt-cands}: Each light purple diamond
corresponds to a domain sampled using our algorithm and each dark
purple diamond corresponds to one sampled using the algorithm of
\citet{szu-boe-bre-fal-nie-sko-sli-tal:j:map} (the diamonds for $3$
and $4$ candidates are much larger as there are very few maximal
single-crossing domains in this case, while remaining ones are very
small, looking rather like small dots than diamonds, but there are
very many of them).  We make two observations: Outer diversity of
single-crossing domains seems to quickly diminish, and outer diversity
of domains sampled using both algorithms seems to follow essentially
the same distribution (Kolmogorov-Smirnov test gives $D \approx 0$ and
$p$-value $\approx 1$).  This is reassuring as for larger domain sizes using
the uniform algorithm is infeasible.
\begin{finding}
  For up to $9$ candidates, the distributions of outer diversity of
  maximal single-crossing domains sampled uniformly at random and
  using the algorithm of \citet{szu-boe-bre-fal-nie-sko-sli-tal:j:map}
  are near-identical.
\end{finding}

\subsection{Single-Peakedness on a Tree}

We conclude this section by considering the family of $\SPT$ domains,
known to be weak Condorcet (the only $\SPT$ Condorcet domain is
Black's $\SP$). In \Cref{fig:sc-spt-cands} we show outer diversity
values of these domains for all nonisomorphic SP trees with
$m \in \{3, \ldots, 9\}$ candidates. Not surprisingly, the smaller is
the diameter of the underlying tree, the higher is the outer diversity
of the domain it generates. Indeed, the highest outer diversity is
achieved by a star graph (where we have one central node and each
other node is connected to it).

\begin{finding}
  For $m \in \{3, \ldots, 9\}$, outer diversities of $\SPT$ domains
  are negatively correlated with the diameter of their underlying trees
  (PCC monotonically goes from $-1$ for $4$ candidates to
  $-0.931$ for $9$ candidates). $\SPT$ domain based on a star graph has
  the highest outer diversity.
\end{finding}

In \Cref{fig:spt6} we plot the six nonisomorphic SP trees for the case
of six candidates, as well as provide the cardinality of the domains
they generate and their outer diversity values. It is interesting how,
except for Black's SP (depicted on the bottom), all these domains have
significantly more votes than even the largest Condorcet domains.
Hence the nature of weak Condorcet domains is quite different from
that of Condorcet ones and understanding them better is an interesting
research challenge.

\section{Asymptotic Behavior of Outer Diversity}

\setlength{\columnsep}{16pt}
\begin{wrapfigure}{r}{0.5\textwidth}
    \vspace{-0.5cm}
    \centering
    \includegraphics[width=0.5\textwidth]{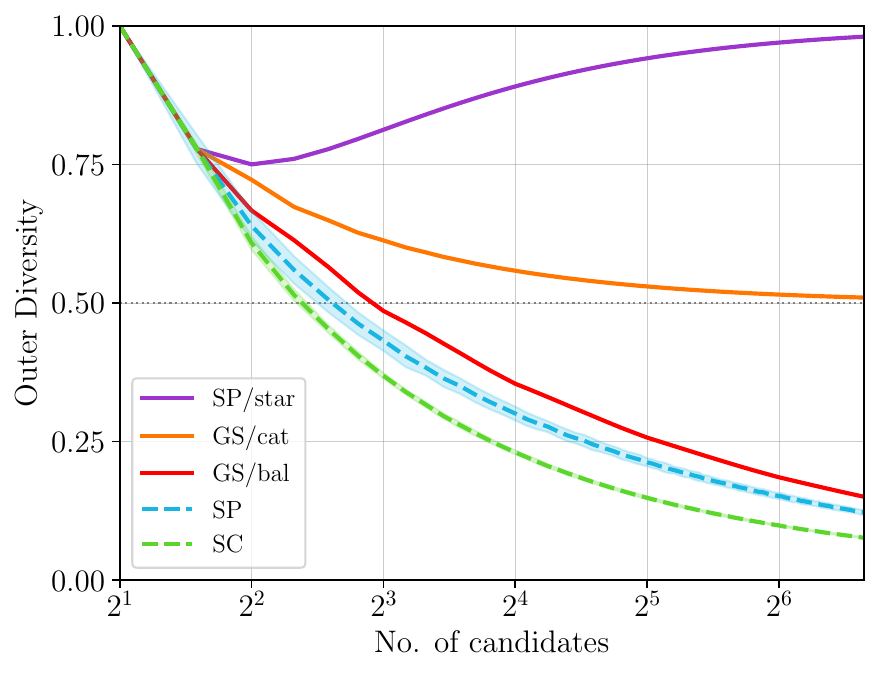}
    \caption{Outer diversity of $\SP$, $\SPstar$, $\GScat$, $\GSbal$, and $\SC$ for
      large numbers of candidates. Note that the $x$ axis is scaled logarithmically.
    Shaded areas show two standard deviations for sampling-based computations.}
    \label{fig:swap}
    \vspace{-0.2cm}
\end{wrapfigure}

\label{sec:asymptotic}
In \Cref{fig:swap}, we plot outer diversity values for $\SC$, $\SP$,
$\SPstar$, $\GScat$, and $\GSbal$, as functions of the number $m$ of
candidates, for $m$ up to $100$; our main goal is to analyze these
values asymptotically. We selected these domains because each of them
is in some way extreme. Indeed, in \Cref{sec:few} we saw that for six
candidates, $\SPstar$, $\GScat$ and $\SP$ have the highest outer
diversity among, respectively, $\SPT$ domains, all Condorcet domains,
and Arrow's SP domains, whereas $\GSbal$ has the lowest outer
diversity among $\GS$ domains.  $\SC$ domains
serve as representatives of Condorcet domains with low outer diversity
that we can sample. Further, they are commonly studied in
computational social choice (as are $\SP$ and $\GS$ domains).

To obtain \Cref{fig:swap}, for $\SP$ and $\SC$ we used
the approximate sampling algorithm provided by
\citet{fal-sor-szu-was:t:outer-diversity}, whereas for $\GScat$,
$\GSbal$, and $\SPstar$ we used our own, exact algorithms
(see \Cref{thm:out-div-gs-cat,thm:outdiv-gst,thm:out-div-sp-star}).
Inspecting
\Cref{fig:swap}, we see that outer diversity of
$\SPstar$ and $\GScat$ seem to approach $1$ and $\nicefrac{1}{2}$,
respectively, whereas outer diversity of $\SC$, $\SP$, and $\GSbal$
might be approaching $0$.  In this section we confirm these limit
values theoretically (except for $\SP$, for which we do not yet have
sufficiently good theoretical understanding), describe our exact
algorithms, and provide some additional results

We first consider the case of $\SC$.
\citet{fal-sor-szu-was:t:outer-diversity} have shown that for every
family of domains, with one domain for each number of candidates,
whose sizes are bounded by a constant, the outer diversity of these
domains tends to $0$ as the number of candidates increases. We extend
this result to domains of polynomial size.  Since every $\SC$ domain
over $m$ candidates contains at most $O(m^2)$ votes, this means that
outer diversity of $\SC$ domains tends to $0$.\footnote{Interestingly,
  this also applies to so-called Euclidean domains of a fixed
  dimension. In such domains each candidate and each voter is
  represented as a point in $\reals^d$, where $d$ is some constant. A
  voter ranks the candidates with respect to their nondecreasing
  distance from his or her point. Each such domain contains
  $O(m^{2d})$ votes.}

\begin{theorem}\label{thm:poly-to-zero}
    Let us fix a polynomial $P(m)$ and 
    let $D_2$, $D_3$, $\ldots$ be a sequence of domains, where each $D_m$ contains at most $P(m)$ rankings over $m$ candidates. Then $\lim_{m \rightarrow \infty}\out(D_m) = 0$.
\end{theorem}
\appendixproof{Theorem}{thm:poly-to-zero}{
    Fix an arbitrary polynomial $P(m)$ and the sequence of domains $D_2$, $D_3$, $\ldots$,
    where each $D_m$ contains at most $P(m)$ rankings over $m$ candidates $C$.
    Let $k$ be the largest power in $P(m)$,
    then $|D_m| = \Theta(m^k)$.
    We assume that $k>0$ and $|D_m|>2$,
    as the case of $k=0$, i.e., constant number of votes in a domain,
    was already considered in~\citep{fal-sor-szu-was:t:outer-diversity}.
    Let us also denote $d_{\max} = \binom{m}{2}$, for convenience.

    Let $\pi_m : \mathcal{L}(C) \rightarrow D_m$ be a function that
    assigns every vote $u \in \mathcal{L}(C)$ a vote from the domain $D_m$
    in a way that minimizes the swap distance, i.e.,
    $\swap(u,\pi_m(u))=\min_{v \in D_m}\swap(u,v).$
    Then, we can write $\ansd$ of $D_m$
    in terms of votes from $\mathcal{L}(C)$ assigned to each vote in $D_m$, i.e.,
    \[
        \ansd(D_m) =
        \frac{\sum_{u \in \mathcal{L}(C)}\min_{v \in D_m} \swap(u,v)}{m! \cdot d_{\max}} =
        \frac{\sum_{v \in D_m}\sum_{u \in \pi^{-1}_m(v)} \swap(u,v)}{m! \cdot d_{\max}}.
    \]

    Let $M_{m,0},M_{m,1},\dots,M_{m,d_{\max}}$ be so called \emph{Mahonian numbers}
    where each $M_{m,i}$ counts the number of votes in $\mathcal{L}(C)$
    that are at swap distance $i$ from a fixed vote $v$.%
    \footnote{See, e.g., \url{https://oeis.org/A008302}.}
    For a number $t \in \mathbb{N}$,
    let $\ell(t)$ be a maximum integer such that
    $\sum_{i=0}^{\ell(t)} M_{m,i} \le t$.
    Also, let $F(t)$ be the minimum sum of swap distances
    from a set of $t$ pairwise distinct votes to a fixed vote $v$, i.e.,
    \[
        F(t) = \sum_{i=0}^{\ell(t)} i\cdot M_{m,i} + \left(t - \sum_{i=0}^{\ell(t)} M_{m,i}\right) \cdot(\ell(t) + 1).
    \]
    Then, since all votes in $\pi^{-1}_m(v)$ for each $v \in D_m$ are pairwise distinct,
    we get that
    \[
        \ansd(D_m) \ge
        \frac{\sum_{v \in D_m}F(|\pi^{-1}_m(v)|)}{m! \cdot d_{\max}}.
    \]
    
    Observe that the value of the above term is minimized,
    if the sets $\pi^{-1}_m(v)$ for different votes $v \in D_m$
    have as similar sizes as possible,
    i.e., each vote from the domain $D_m$ is assigned
    an (almost) equal number of votes from $\mathcal{L}(C)$.
    Hence, we can provide its lower bound by exchanging $F(|\pi^{-1}_m(v)|)$
    for $F(\lfloor m!/|D_m| \rfloor)$, which yields
    \begin{align}
    \notag
        \ansd(D_m) &\ge
        \frac{\sum_{v \in D_m} F(\lfloor m!/|D_m| \rfloor)}{m! \cdot d_{\max}} \\
    \notag
        &= \frac{|D_m|}{m!} \cdot \left\lfloor \frac{m!}{|D_m|}\right\rfloor \cdot
        \frac{1}{d_{\max}} \cdot 
        \frac{F(\lfloor m!/|D_m| \rfloor)}{\lfloor m!/|D_m| \rfloor}\\
    \notag
        &\ge \frac{|D_m|}{m!} \cdot \left(  \frac{m!}{|D_m|} - 1\right) \cdot
        \frac{1}{d_{\max}} \cdot 
        \frac{F(\lfloor m!/|D_m| \rfloor)}{\lfloor m!/|D_m| \rfloor}\\
    \label{eq:poly-to-zero:ansd1}
        &= \left(1 - \frac{\Theta(m^k)}{m!}\right) \cdot
        \frac{1}{d_{\max}} \cdot  \frac{F(\lfloor m!/|D_m| \rfloor)}{\lfloor m!/|D_m| \rfloor}.
    \end{align}

    Now, let $V$ be an arbitrary set of $\lfloor m!/|D_m| \rfloor$
    closest votes from $\mathcal{L}(C)$ to a given fixed vote $v$.
    Also, let $x$ be a vote chosen uniformly at random from $\mathcal{L}(C)$.
    Then, observe that we can express
    $F(\lfloor m!/|D_m| \rfloor)/\lfloor m!/|D_m| \rfloor$
    in terms of expected swap distance between $x$ and $v$
    conditioned on the fact that $x$ belongs to $V$, i.e.,
    \[ 
        \frac{F(\lfloor m!/|D_m| \rfloor)}{\lfloor m!/|D_m| \rfloor} =
        \mathbb{E}(\swap(x,v) | x \in V).
    \]
    Next, let us denote $\ell_m \coloneqq \ell(\lfloor m!/|D_m| \rfloor)$
    and observe that
        \[
        \mathbb{E}(\swap(x,v) | x \in V) \ge \mathbb{E}(\swap(x,v) | \swap(x,v) \le \ell_m), 
    \]
    as the set $\{ x \in \mathcal{L}(C) : \swap(x,v) \le \ell_m\}$ is a subset of $V$,
    and the expected value can only decrease as we remove the votes from $V$
    with the highest swap distance to $v$.
    Combining this with Inequality~\eqref{eq:poly-to-zero:ansd1},
    we obtain
    \begin{equation}
    \label{eq:poly-to-zero:expected-distance-bound}
        \ansd(D_m) \ge
        \left(1 - \frac{\Theta(m^k)}{m!}\right) \cdot
        \frac{1}{d_{\max}} \cdot
        \mathbb{E}(\swap(x,v) | \swap(x,v) \le \ell_m).
    \end{equation}

    In order to bound the expected conditional distance,
    we will first look at the asymptotics of $\ell_m$.
    To this end, we will use the \emph{Lehmer code} and the \emph{Hoeffding's inequality}.

    First, however, let us show that
    \begin{equation}
    \label{eq:poly-to-zero:ell_m_less_than_half}
        \ell_m < d_{\max}/2.
    \end{equation}
    Assume otherwise, i.e., $\ell_m \ge d_{\max}/2$.
    We use that fact that the Mahonian numbers are symmetric,
    i.e., $M_{m,i} = M_{m,d_{\max}-i}$ for each $i \in \{0,1,\dots,d_{\max}\}$.
    This implies that
    \[
        \sum_{i=0}^{\ell_m} M_{m,i} \ge 
        \frac{1}{2} \sum_{i=0}^{d_{\max}} M_{m,i} =
        \frac{m!}{2}.
    \]
    On the other hand, from the definition of $\ell_m$,
    \[
        \sum_{i=0}^{\ell_m} M_{m,i} \le \left\lfloor \frac{m!}{|D_m|} \right\rfloor  \le \frac{m!}{|D_m|}.
    \]
    Combining both inequalities together, we get that $|D_m|\le 2$,
    which contradicts our assumption.
    Thus, Inequality~\eqref{eq:poly-to-zero:ell_m_less_than_half} holds.
    
    Now, let us denote the linear order in vote $v$ by
    $v : c_1 \succ c_2 \succ \dots \succ c_m$.
    The Lehmer code of vote $x \in \mathcal{L}(C)$
    with respect to the reference vote $v$
    is a sequence $(\sigma(x,v,1), \sigma(x,v,2),\dots,\sigma(x,v,m))$
    in which $\sigma(x,v,i)$ denotes the number of candidates
    ranked by $v$ below $c_i$
    which $x$ ranks above $c_i$, i.e.,
    \[
        \sigma(x,v,i) = |\{ c_j \in C : c_i \succ_v c_j \land c_j \succ_x c_i \}|. 
    \]
    The Lehmer code has three properties that are important for us.
    First, it holds that the sum of Lehmer codes is equal to the swap distance, i.e.,
    \[
        \swap(x,v) = \sum_{i=1}^m \sigma(x,v,i).
    \]
    Second, for every $i \in [m]$,
    it holds that $\sigma(x,v,i) \in \{0, 1,\dots, m-i\}$.
    Third, every sequence $(s_1,s_2,\dots,s_m)$
    such that $s_i \in \{0, 1,\dots, m-i\}$
    uniquely determines vote $x$ such that
    $(\sigma(x,v,i))_{i \in [m]} = (s_i)_{i \in [m]}$.
    Thus, in order to pick $x$ from $\mathcal{L}(C)$ uniformly at random
    we can alternatively pick its Lehmer code uniformly at random.
    In other words, we can express a random variable $\swap(x,v)$
    as a sum of $m$ independent random variables $\sigma(x,v,1),\dots,\sigma(x,v,m)$,
    where for each $i \in [m]$ we have that
    $\sigma(x,v,i)$ follows a uniform distribution from the set
    $\{0,1,\dots,m-i\}$.

    This allows us to use the Hoeffding's inequality,
    which in general says (in its lower tail version) that
    for independent random variables $X_1,\dots,X_n$
    such that each $X_i$ takes values from $[a_i,b_i]$
    and $t \ge 0$
    it holds that
    \[
        \textstyle
        \mathbb{P}\left(
            \sum_{i\in [n]}(X_i) - \mathbb{E}\left(\sum_{i\in [n]}X_i \right) \le -t
        \right) \le 
        \exp\left(- \frac{2t^2}{\sum_{i\in [n]}(b_i - a_i)^2}\right).
    \]
    By applying it to our setting
    and taking
    $t = \mathbb{E}(\swap(x,v)) - (\ell_m + 1) = d_{\max}/2 - (\ell_m + 1)$,
    for which we have $t \ge 0$ by Inequality~\eqref{eq:poly-to-zero:ell_m_less_than_half},
    we obtain
    \[
        \textstyle
        \mathbb{P}\left(
            \swap(x,v) - d_{\max}/2 \le (\ell_m + 1) - d_{\max}/2
        \right) \le 
        \exp\left(- \frac{2(\ell_m + 1 - d_{\max}/2)^2}{\sum_{i\in [m]}(m-i)^2}\right).
    \]
    Clearly,
    \(
        \mathbb{P}\left(
            \swap(x,v) - d_{\max}/2 \le (\ell_m + 1) - d_{\max}/2
        \right) =
        \mathbb{P}(\swap(x,v) \le \ell_m + 1),
    \)
    thus we get
    \[
        \textstyle
        \mathbb{P}\left(
            \swap(x,v) \le \ell_m + 1
        \right) \le 
        \exp\left(- \frac{2(d_{\max}/2 - (\ell_m + 1))^2}{\sum_{i\in [m]}(m-i)^2}\right).
    \]
    On the other hand, observe that
    \begin{equation}
    \label{eq:poly-to-zero:swap-prob-bound}
        \mathbb{P}\left( \swap(x,v) \le \ell_m + 1 \right) =
        \frac{1}{m!} \sum_{i=0}^{\ell_m + 1} M_{m,i} \ge
        \frac{1}{m!} \cdot \frac{m!}{|D_m|} = \frac{1}{|D_m|} = \Theta(m^{-k}),
    \end{equation}
    where the inequality comes from the fact that
    by definition $\ell_m$
    is the largest natural number such that
    $\sum_{i=0}^{\ell_m} M_{m,i} \le \lfloor m!/|D_m| \rfloor$,
    hence if we extend the sum to $\ell_m + 1$
    it will exceed $m!/|D_m|$, as $M_{m,i} > 1$ for $i > 0$.
    Thus, we get
    \[
        \textstyle
        \Theta(m^{-k}) \le 
        \mathbb{P}\left(
            \swap(x,v) \le \ell_m + 1
        \right) \le 
        \exp\left(- \frac{2(d_{\max}/2 - (\ell_m + 1))^2}{\sum_{i\in [m]}(m-i)^2}\right).
    \]
    On each side of the above inequality let us take the natural logarithm and multiply by $-1$,
    then we obtain
    \[
        \Theta(k\log(m)) \ge \frac{2(d_{\max}/2 - (\ell_m + 1))^2}{\sum_{i\in [m]}(m-i)^2}.
    \]
    Since $k$ is a constant, we can remove it from the $\Theta$ notation,
    as well as constant $2$ from the right hand side. Then,
    \[
        \Theta(\log(m)) \cdot \sum_{i\in [m]}(m-i)^2 \ge (d_{\max}/2 - (\ell_m + 1))^2.
    \]
    By the fact that $\sum_{i\in [m]}(m-i)^2 = (m-1)m(2m-1)/6$, we get
    \[
        \Theta(m^3\log(m)) \ge (d_{\max}/2 - (\ell_m + 1))^2.
    \]
    Since we know that $d_{\max}/2 - (\ell_m + 1) \ge 0$ (by Inequality~\eqref{eq:poly-to-zero:ell_m_less_than_half}), we can take the square root to obtain
    \[
        \Theta\left(m^{3/2}\textstyle\sqrt{\log(m)}\right) \ge d_{\max}/2 - (\ell_m + 1).
    \]
    And from this we obtain the asymptotic bound on $\ell_m$ that is
    \begin{equation}
    \label{eq:poly-to-zero:ell-asymptotics}
        \textstyle
        \ell_m \ge d_{\max}/2 - O\left(m^{3/2}\sqrt{\log(m)}\right).
    \end{equation}

    Now, observe that we could repeat our reasoning for an arbitrary
    sequence of domains $D^{\kappa}_m$
    of size $|D^\kappa_m| = m^\kappa$ and obtain bounds analogous
    to~\eqref{eq:poly-to-zero:swap-prob-bound} and~\eqref{eq:poly-to-zero:ell-asymptotics}
    for different $\ell^\kappa \coloneqq \ell(\lfloor m! /|D|\rfloor)$.
    Let $k' \in \mathbb{N}$ be a smallest integer such that $k' \ge 2k$
    and $\ell' \coloneqq \ell^{k'} < \ell_m$,
    and $k'' \in \mathbb{N}$ be a smallest integer such that $k'' \ge 2k'$
    and $\ell'' \coloneqq \ell^{k''} < \ell'$,
    Then,
    \[
        \ell'' \ge d_{\max}/2 - O\left(m^{3/2}\textstyle\sqrt{\log(m)}\right)
        \quad \mbox{and} \quad
        \mathbb{P}\left(\swap(x,v) \le \ell'+ 1\right) \ge \Theta(m^{-k'}).
    \]
    Finally, observe that by definition of $\ell(t)$,
    we obtain
    \[
        \mathbb{P}\left( \swap(x,v) \le \ell'' \right) =
        \frac{1}{m!} \sum_{i=0}^{\ell''} M_{m,i} \le
        \frac{1}{m!} \cdot \left\lfloor\frac{m!}{|D^{k''}_m|}\right\rfloor \le \frac{1}{|D^{k''}_m|} = \Theta(m^{-k''}).
    \]

    Then, from Inequality~\eqref{eq:poly-to-zero:expected-distance-bound}
    we obtain
    \begin{align*}
        \ansd(D_m) &\ge 
        \frac{1}{d_{\max}} \cdot 
        \mathbb{E}(\swap(x,v) | \swap(x,v) \le \ell_m) 
        \left(1 - \frac{\Theta(m^k)}{m!}\right)\\
        &\ge 
        \frac{\ell''+1}{d_{\max}} 
        \mathbb{P}(\swap(x,v) \ge \ell'' +1 |\swap(x,v) \le \ell_m)
        \left(1 - \frac{\Theta(m^k)}{m!}\right)\\
        &\ge 
        \frac{\ell''+1}{d_{\max}} 
        \mathbb{P}(\swap(x,v) \ge \ell'' +1 |\swap(x,v) \le \ell' + 1)
        \left(1 - \frac{\Theta(m^k)}{m!}\right)\\
        &= 
        \frac{\ell''+1}{d_{\max}} 
        \left( 1 -  \mathbb{P}(\swap(x,v) \le \ell'' +1 |\swap(x,v) \le \ell' + 1) \right)
        \left(1 - \frac{\Theta(m^k)}{m!}\right)\\
        &=
        \frac{\ell''+1}{d_{\max}} 
        \left( 1 -  \frac{\mathbb{P}(\swap(x,v) \le \ell'' +1)}{\mathbb{P}(\swap(x,v) \le \ell' + 1)} \right)
        \left(1 - \frac{\Theta(m^k)}{m!}\right)\\
        &\ge
        \frac{d_{\max}/2 - O\left(m^{3/2}\sqrt{\log(m)}\right)}{d_{\max}} 
        \left( 1 -  \frac{\Theta(m^{-k''})}{\Theta(m^{-k'})}\right)
        \left(1 - \frac{\Theta(m^k)}{m!}\right)\\
        &=
        \frac{1}{2} \left(1 - O\left(\sqrt{\log(m)/m}\right)\right) 
        \left( 1 -  \Theta\left(m^{-(k'' - k')}\right)\right)
        \left(1 - \frac{\Theta(m^k)}{m!}\right).\\
    \end{align*}
Therefore,
\(
    \lim_{m \to \infty} \ansd(D_m) \ge \frac{1}{2}.
\)
Since $\ansd$ can never be greater than half,
we get that
\(
    \lim_{m \to \infty} \ansd(D_m) = \frac{1}{2},
\)
which concludes the proof.
}

\begin{corollary}
  For every $m \in \mathbb{N}$, let $\SC_m$ be a single-crossing domain
  over $m$ candidates. Then
  $\lim_{m \rightarrow \infty} \out(\SC_m) = 0$.
\end{corollary}

Next, we move to the case of $\GS$.  In \Cref{thm:outdiv-gst}, we show
how to efficiently compute the outer diversity of any $\GS$ domain
specified by a binary tree. The algorithm aggregates votes according
to the positions of one child block at each node, computes the
corresponding local contribution, and combines it with the
contributions of the two subtrees in a bottom-up way.

\begin{theorem}\label{thm:outdiv-gst}
  For any domain $\GS(T)$ over $m$ candidates, where $T$ is a given binary tree,
  there exists an $O(m^4)$-time algorithm that computes $\outdiv(\GS(T))$.
\end{theorem}
\begin{proof}[Proof sketch of \Cref{thm:outdiv-gst}]
  We defer the full formal proof to the appendix
  and describe here the main idea.
  To compute $\outdiv(\GS(T))$,
  where $T$ is labeled with candidate set $C$ and $|C|=m$,
  it is enough to compute the sum of swap distances
  from all votes from $\calL(C)$ to the domain $\GS(T)$.
  We first consider the distance from a fixed vote $u \in \calL(C)$ to $\GS(T)$.
  We argue that this distance decomposes over the nodes of the tree.
  For an internal node $x$,
  the local cost of $u$ at $x$ is the cost of arranging the two child blocks
  in one of the two admissible orders,
  while the remaining cost is paid independently inside the two subtrees.
  In \Cref{lem:fxsu-formula} we show that,
  if the left child of $x$ contains $t$ candidates,
  then this local cost depends only on the sum of the positions occupied
  in $u$ by these $t$ candidates.
  Thus, when summing over all votes,
  we do not need to consider all $m!$ votes separately:
  we can group together votes with the same value of this position sum.
  By grouping votes in this way, \Cref{lem:fx-formula} gives an aggregated
  local contribution $F(x)$ for each node $x$.

  To compute these aggregated contributions,
  we use a dynamic-programming table $A(m,t,s)$.
  The value $A(m,t,s)$ counts the number of sets of $t$ positions among
  $\{1,\dots,m\}$ whose indices sum to $s$.
  For each such set of positions,
  the candidates inside the two child blocks can be permuted independently.
  Hence $A(m,t,s)$ lets us count how many votes have each type,
  and therefore compute the value $F(x)$ without iterating over all votes.
  Once all local contributions are known,
  we process the tree bottom-up.
  At every internal node, the total value for the subtree is obtained by
  adding the local contribution $F(x)$ and appropriately scaled contributions
  of the two children.
  This recurrence is presented in \Cref{lem:sx-recurrence}.

  Finally, the table $A$ has $O(m^4)$ states,
  since the sum parameter is bounded by $O(m^2)$.
  All remaining computations, including the computation of all values $F(x)$
  and the bottom-up computation of the subtree values,
  are dominated by this dynamic program.
  Hence $\outdiv(\GS(T))$ can be computed in time $O(m^4)$.
\end{proof}
\appendixproof{Theorem}{thm:outdiv-gst}{
Let $C = \{c_1,c_2,\dots,c_m\}$ be the set of candidates
and let $T_x$ denote the subtree rooted at a node $x$ of $T$.
For each node $x$,
let $C_x$ be the set of candidates labeling the leaves of $T_x$.
Let us define $S(x) \defeq \sum_{u \in \calL(C_x)} \swap(u,\GS(T_x))$.
From the definitions of $\ansd(\cdot)$ and $\outdiv(\cdot)$,
and since $C_\roott = C$ for the root $\roott$ of $T$, we have:
\begin{align}
  \outdiv(\GS(T))
  = 1 - \frac{2 \cdot \sum_{u \in \calL(C)} \swap(u,\GS(T))}{\binom{m}{2} \cdot m!}
  = 1 - \frac{2 \cdot S(\roott)}{\binom{m}{2} \cdot m!}.\label{eq:outdiv-gst}
\end{align}
Therefore, in order to prove the theorem statement
it suffices to compute $S(\roott)$ in time $O(m^4)$.

We start with an auxiliary dynamic programming table.
Let us define a function $A(m,t,s)$ which represents
how many subsets of $t$ positions from the set of positions
$\{1,2,3,\dots,m\}$ achieve the sum of indices equal to $s$.
Formally,
\[
  A(m,t,s) \defeq
  |\{ X \subseteq \{1,2,\dots,m\} : |X|=t, \textstyle\sum_{x \in X} x = s \}|.
\]
We can obtain this number using a basic recurrence formula based on the fact that
either we use a subset of positions from $\{1,2,\dots,m-1\}$
or we add $m$ to a subset of $t-1$ positions; hence:
$A(m,t,s) = A(m-1,t,s)+A(m-1,t-1,s-m)$.
Additionally we define boundary conditions:
$A(m,0,0) = 1$ for every $m \geq 0$
(since the empty set satisfies the condition), and
$A(m,t,s) = 0$ whenever $m < 0$, or $t < 0$, or $s < 0$.

Let us now consider an arbitrary internal node $x$ of $T$.
Let $\ell$ and $r$ be the left and right child of $x$, respectively.
Let $m_x = |C_x|$ and $t = |C_\ell|$.
For two disjoint sets of candidates $X,Y \subseteq C_x$ and
a vote $u \in \calL(C_x)$, we write $X \succ_u Y$
if every candidate from $X$ is ranked above every candidate from $Y$ in~$u$.

We define a domain
$D_x \defeq \{u \in \calL(C_x) : C_\ell \succ_u C_r \text{ or } C_r \succ_u C_\ell\}$, i.e.,
$D_x$ consists of all votes $u$ such that
either $u$ ranks all candidates from $C_\ell$ before any candidate from $C_r$
or $u$ ranks all candidates from $C_r$ before any candidate from $C_\ell$.
Since $x$ is the root of $T_x$,
every vote $u \in \GS(T_x)$ satisfies
either $C_\ell \succ_u C_r$ or $C_r \succ_u C_\ell$.
Hence, $\GS(T_x) \subseteq D_x$.

For every vote $u \in \calL(C_x)$,
let $s_x(u) \defeq \sum_{i=1}^t p_i$,
where $1 \leq p_1 < p_2 < \cdots < p_t \leq m_x$
are the positions occupied in $u$ by the candidates from $C_\ell$.
The following lemma (\Cref{lem:fxsu-formula}) gives a closed formula for
$\swap(u,D_x)$,
which we call \emph{the local cost of $u$ at node $x$}.

\begin{lemma}\label{lem:fxsu-formula}
  For every $u \in \calL(C_x)$, we have
  \[
    \swap(u,D_x) = \min \left\{
      s_x(u) - \textstyle\frac{t(t+1)}{2},\quad
      t \cdot (m_x-t) + \textstyle\frac{t(t+1)}{2} - s_x(u)
    \right\}.
  \]
\end{lemma}
\begin{proof}[Proof of \Cref{lem:fxsu-formula}]
  We fix $u \in \calL(C_x)$.
  Let $1 \leq p_1 < p_2 < \cdots < p_t \leq m_x$
  be the positions occupied in $u$ by the candidates from $C_\ell$.

  We observe that the minimum number of adjacent swaps required to transform $u$ so that
  the candidates from $C_\ell$ occupy the first $t$ positions equals
  $\sum_{i=1}^t (p_i -  i) = s_x(u) - \frac{t(t+1)}{2} \eqdef \alpha$.
  Indeed, this quantity is exactly the number of pairs $(a,b)$ such that
  $a \notin C_\ell$, $b \in C_\ell$,
  and $u \colon a \pref b$ ($a$ is ranked higher than $b$ in $u$).
  For the candidate from $C_\ell$ occupying position $p_i$,
  there are exactly $p_i-i$ such candidates above it,
  since among the first $p_i-1$ positions
  exactly $i-1$ are occupied by candidates from $C_\ell$.
  Each adjacent swap changes the number of such pairs by at most one,
  and the target arrangement is obtained exactly when no such pair remains.
  Hence the above quantity is the minimum number of adjacent swaps.

  Similarly, the minimum number of adjacent swaps required to transform $u$
  so that the candidates from $C_\ell$ occupy the last $t$ positions equals
  $\sum_{i=1}^t \bigl( (m_x-t+i) - p_i \bigr)
  = t \cdot (m_x-t) + \frac{t(t+1)}{2} - s_x(u) \eqdef \beta$.

  Every transformation of $u$ into a vote from $D_x$
  must make the candidates from $C_\ell$ occupy either the first $t$ positions
  or the last $t$ positions.
  Hence $\swap(u,D_x) \ge \min \{ \alpha, \beta \}$.
  This minimum is also enough to transform $u$ into a vote from $D_x$, so
  $\swap(u,D_x) \le \min \{ \alpha, \beta \}$.
  This finishes the proof of \Cref{lem:fxsu-formula}.
\end{proof}

Let $F(x)$ denote the sum of the local costs of all votes from
$\calL(C_x)$ at node $x$, i.e.,
$F(x) \defeq \sum_{u \in \calL(C_x)} \swap(u,D_x)$.
The following lemma (\Cref{lem:fx-formula}) shows that
we can compute $F(x)$ more efficiently than by summing $\swap(u,D_x)$
over all votes from $\calL(C_x)$ one by one.
Instead, we group together all votes for which the candidates from $C_\ell$
occupy sets of positions with the same sum $s$,
and count how many such votes there are.
We define
$f_x(s) \defeq \min \{
  s - \frac{t(t+1)}{2},\quad
  t \cdot (m_x-t) + \frac{t(t+1)}{2} - s
\}$.

\begin{lemma}\label{lem:fx-formula}
  We have
  $F(x) = \sum_s f_x(s) \cdot A(m_x,t,s)\cdot t!\cdot (m_x-t)!$.
\end{lemma}
\begin{proof}[Proof of \Cref{lem:fx-formula}]
  For a fixed value of $s$,
  there are exactly $A(m_x,t,s)$ sets of $t$ positions in
  $\{1,2,\dots,m_x\}$ whose indices sum to $s$.
  For each such set of positions,
  the candidates from $C_\ell$ can be arranged within these positions in $t!$ ways,
  and the candidates from $C_r$ can be arranged in the remaining $m_x-t$ positions in $(m_x-t)!$ ways.
  Hence, the number of votes $u \in \calL(C_x)$
  for which the candidates from $C_\ell$ occupy a set of positions
  whose indices sum to $s$ is equal to
  $A(m_x,t,s) \cdot t! \cdot (m_x-t)!$.
  By \Cref{lem:fxsu-formula}, for each such vote $u$,
  $\swap(u,D_x) = f_x(s)$.
  Summing over all values of $s$, we obtain
  $F(x) = \sum_s f_x(s) \cdot A(m_x,t,s)\cdot t!\cdot (m_x-t)!$.
  Note that only values of $s$ in the range
  $\frac{t(t+1)}{2} \leq s \leq m_x t -\frac{t(t-1)}{2}$
  contribute nonzero terms.
  This finishes the proof of \Cref{lem:fx-formula}.
\end{proof}

For a vote $u \in \calL(C_x)$ and a set $X \subseteq C_x$,
let $u|_X$ denote the restriction of $u$ to $X$, i.e.,
the unique vote in $\calL(X)$ such that,
$a \pref_{u|_X} b \iff a \pref_u b$
for every $a,b \in X$.

The following lemma gives a recurrence formula for computing $S(x)$.

\begin{lemma}\label{lem:sx-recurrence}
  Let $x$ be a node of $T$.
  If $x$ is a leaf, then $S(x)=0$.
  If $x$ is an internal node with children $\ell$ and $r$, then
  \[ 
    S(x) = F(x)
         + \binom{m_x}{t} \cdot (m_x-t)! \cdot S(\ell)
         + \binom{m_x}{m_x-t} \cdot t! \cdot S(r).
  \]
\end{lemma}
\begin{proof}[Proof of \Cref{lem:sx-recurrence}]
  If $x$ is a leaf, then $|C_x|=1$,
  hence $\calL(C_x)$ consists of a single vote,
  which already belongs to $\GS(T_x)$.
  Therefore, $S(x)=0$.

  Assume now that $x$ is an internal node with children $\ell$ and $r$.
  Fix a vote $u \in \calL(C_x)$.
  Transforming $u$ into a vote from $\GS(T_x)$ is equivalent to performing
  three independent tasks:
  \begin{enumerate}
   \item arranging the two blocks $C_\ell$ and $C_r$
   so as to obtain a vote from $D_x$,
   \item transforming $u|_{C_\ell}$ into a vote from $\GS(T_\ell)$,
   \item transforming $u|_{C_r}$ into a vote from $\GS(T_r)$.
  \end{enumerate}
  Indeed, every vote from $\GS(T_x)$ satisfies these three conditions.
  Since the corresponding swaps concern disjoint pairs of candidates,
  their costs add up.
  Conversely, once we choose one of the two admissible orders of the blocks
  and arbitrary votes from $\GS(T_\ell)$ and $\GS(T_r)$,
  they uniquely determine a vote from $\GS(T_x)$.
  Therefore,
  \begin{align}
    \swap(u,\GS(T_x)) =
      \swap(u, D_x)
    + \swap(u|_{C_\ell}, \GS(T_\ell))
    + \swap(u|_{C_r}, \GS(T_r)).\label{eq:swap-u-gstx}
  \end{align}

  This allows us to calculate
  $S(x) = \sum_{u \in \calL(C_x)} \swap(u,\GS(T_x))$
  by splitting $\swap(u,\GS(T_x))$ into 3 parts.
  The first part is
  $\sum_{u \in \calL(C_x)} \swap(u, D_x) = F(x)$
  by the definition of $F(x)$.

  For the second part,
  we group together all votes $u \in \calL(C_x)$ that induce the same
  restriction $u|_{C_\ell}$.
  For every fixed vote $w \in \calL(C_\ell)$,
  the value $\swap(w,\GS(T_\ell))$ appears once for each vote
  $u \in \calL(C_x)$ such that $u|_{C_\ell}=w$.
  The number of such votes is exactly
  $\binom{m_x}{t} \cdot (m_x-t)!$:
  we choose the $t$ positions occupied by the candidates from $C_\ell$,
  and then we place the candidates from $C_r$ arbitrarily on the remaining positions.
  Therefore, using the definition of $S(\ell)$, we obtain
  \begin{align*}
    \sum_{u \in \calL(C_x)} \swap(u|_{C_\ell},\GS(T_\ell))
    &= \binom{m_x}{t} \cdot (m_x-t)!
      \cdot \sum_{w \in \calL(C_\ell)} \swap(w,\GS(T_\ell)).\\
    &= \binom{m_x}{t} \cdot (m_x-t)! \cdot S(\ell).
  \end{align*}

  Symmetrically, in the third part,
  we group together all votes $u \in \calL(C_x)$ that induce the same
  restriction $u|_{C_r}$.
  By an analogous analysis, we obtain
  $\sum_{u \in \calL(C_x)} \swap(u|_{C_r},\GS(T_r))
    = \binom{m_x}{m_x-t} \cdot t! \cdot S(r)$.
  Combining the three contributions gives the required recurrence.
\end{proof}

It remains to estimate the running time.
The values $A(m,t,s)$ can be computed by dynamic programming
(using the given recurrence).
For each fixed $m$, the parameter $t$ ranges from $0$ to $m$,
while the parameter $s$ ranges up to $O(m^2)$,
since it is the sum of $t$ distinct numbers from $\{1,2,\dots,m\}$.
Thus, the number of states of the table $A$ is $O(m^4)$,
and each state is computed in constant time from the recurrence.

After precomputing factorials and binomial coefficients in $O(m^2)$ time,
each value $F(x)$ can be computed from \Cref{lem:fx-formula}
by summing over all relevant values of $s$.
Since the range of $s$ has size $O(m_x^2)$,
this takes $O(m_x^2)$ time for a fixed node $x$.
Therefore, all values $F(x)$ can be computed in $O(m^3)$ time in total.

Finally, once all values $F(x)$ are known,
the values $S(x)$ for all nodes $x$ of the tree can be computed
in linear time by processing the nodes bottom-up and
applying \Cref{lem:sx-recurrence}.

Therefore, $S(\roott)$ can be computed in time $O(m^4)$, and consequently
$\outdiv(\GS(T))$ can also be computed in time $O(m^4)$.
}

It turns out that there is a simple modification of the above algorithm
so that instead of computing outer diversity of a $\GS$ domain
specified by a given binary tree $T$, it computes the tree that leads
to either the highest or the lowest outer diversity
(see \Cref{thm:most-div-gs}).
The idea is to choose the optimal split at the root.
If the two subtrees have sizes $a$ and $m-a$,
then the contribution of the root depends only on this split,
while the two subtrees can be optimized independently.
Thus, we compute optimal trees of sizes $1,2,\dots,m$ in increasing order:
for each size, we try all possible splits at the root and
combine the best already-computed subtrees.
The same dynamic program can be used for both maximizing and
minimizing the outer diversity.
Using this algorithm, we extend \Cref{fin:gs}:

\begin{finding}
  $\GScat$ and $\GSbal$ have, respectively, the highest and the lowest
  outer diversity among $\GS$ domains with up to $100$ candidates.
\end{finding}

\begin{theorem}\label{thm:most-div-gs}
  Let $\mathcal{T}_m$ denote the family of all binary trees with $m$ leaves.
  There exists an $O(m^4)$-time algorithm that finds a tree in
  $\arg\max_{T \in \mathcal{T}_m} \outdiv(\GS(T))$.
  Moreover, the same approach can be used to find a tree in
  $\arg\min_{T \in \mathcal{T}_m} \outdiv(\GS(T))$.
\end{theorem}
\appendixproof{Theorem}{thm:most-div-gs}{
  For a tree $T$ whose leaves are labeled by candidates from $C$,
  we write $S(T)$ for the total swap distance from all votes in $\calL(C)$
  to the domain $\GS(T)$, that is,
  \[
    \textstyle S(T) \defeq \sum_{u \in \calL(C)} \swap(u,\GS(T)).
  \]
  This is the same quantity as $S(\roott)$ in the proof of
  \Cref{thm:outdiv-gst},
  where $\roott$ is the root of $T$.
  By \eqref{eq:outdiv-gst}, for every fixed $m$,
  maximizing $\outdiv(\GS(T))$ over $\mathcal{T}_m$ is equivalent to
  minimizing $S(T)$,
  and minimizing $\outdiv(\GS(T))$ is equivalent to maximizing $S(T)$.

  We first describe how to minimize $S(T)$ for a fixed number $m$ of leaves.
  For that we define
  \begin{align}
    M_{\min}(a) \defeq \min_{T \in \mathcal{T}_a} S(T).\label{eq:def-m-min-a}
  \end{align}
  In particular, any tree $T \in \mathcal{T}_m$ achieving $M_{\min}(m)$
  minimizes $S(T)$,
  and therefore maximizes $\outdiv(\GS(T))$.

  Fix a binary tree $T$ and an internal node $x$ of $T$.
  By \Cref{lem:sx-recurrence},
  once the sizes of the two subtrees of $x$ are fixed,
  the local contribution of $x$ to $S(T)$ is independent of the shapes
  of these subtrees.
  More precisely, by \Cref{lem:fx-formula},
  this local contribution depends only on the number $a$ of leaves below $x$
  and on the number of leaves $t$ in the left subtree of $x$.
  Hence, for $a \in \{2,3,\dots,m\}$ and $t \in \{1,2,\dots,a-1\}$,
  we define $F(a,t)$ as the value $F(x)$ from \Cref{lem:fx-formula}
  at any node $x$ with $m_x=a$ and $|C_\ell|=t$.

  For $a \in \{2,3,\dots,m\}$ and $t \in \{1,2,\dots,a-1\}$, we define
  \begin{align}
    M_{\min}(a,t)
    \defeq F(a,t)
      + \binom{a}{t} \cdot (a-t)! \cdot M_{\min}(t)
      + \binom{a}{a-t} \cdot t! \cdot M_{\min}(a-t).\label{eq:def-m-min-a-t}
  \end{align}
  In other words,
  $M_{\min}(a,t)$ is the value of $S(T)$ for a tree $T$
  whose root split has sizes $t$ and $a-t$,
  and whose two subtrees achieve $M_{\min}(t)$ and $M_{\min}(a-t)$,
  respectively.

  The next lemma shows that $M_{\min}(a)$ is obtained by minimizing
  $M_{\min}(a,t)$ over all possible root splits $t$.
  This recurrence will be used by the dynamic program.
  \begin{lemma}\label{lem:mmin-gs-recurrence}
    We have $M_{\min}(1)=0$ and, for every $a \geq 2$,
    \begin{align}
      M_{\min}(a) = \min_{1 \leq t \leq a-1} M_{\min}(a,t).\label{eq:mmin-gs-recurrence}
    \end{align}
  \end{lemma}
  \begin{proof}[Proof of \Cref{lem:mmin-gs-recurrence}]
    The base case $M_{\min}(1)=0$ is immediate.

    We prove the recurrence by two inequalities.
    First, for any fixed $t \in \{1,2,\dots,a-1\}$ take a tree
    $T_t \in \arg\min_{T \in \mathcal{T}_t} S(T)$
    and a tree
    $T_{a-t} \in \arg\min_{T \in \mathcal{T}_{a-t}} S(T)$.
    We join them under a new root,
    with $T_t$ as the left subtree and $T_{a-t}$ as the right subtree,
    and call the resulting tree $T_{t,a-t}$.
    By definition, $S(T_t) = M_{\min}(t)$ and $S(T_{a-t}) = M_{\min}(a-t)$,
    hence
    \begin{align}
      S(T_{t,a-t})
      &\stackrel{\text{\Cref{lem:sx-recurrence}}}{=}
        F(a,t)
        + \binom{a}{t} \cdot (a-t)! \cdot S(T_t)
        + \binom{a}{a-t} \cdot t! \cdot S(T_{a-t})\nonumber\\
      &\stackrel{\phantom{\text{\Cref{lem:sx-recurrence}}}}{=}
        F(a,t)
        + \binom{a}{t} \cdot (a-t)! \cdot M_{\min}(t)
        + \binom{a}{a-t} \cdot t! \cdot M_{\min}(a-t)\nonumber\\
      &\hspace{12pt}\stackrel{\eqref{eq:def-m-min-a-t}}{=}\hspace{12pt}
        M_{\min}(a,t).\label{eq:s-of-t-split-to-t-and-a-minus-t}
    \end{align}
    Since $T_{t,a-t} \in \mathcal{T}_a$ for every $t \in \{1,2,\dots,a-1\}$,
    we have
    \begin{align*}
      M_{\min}(a) \leq S(T_{t,a-t})
      \stackrel{\eqref{eq:s-of-t-split-to-t-and-a-minus-t}}{=}
      M_{\min}(a,t) \quad \forall t \in \{1,2,\dots,a-1\}.
    \end{align*}
    Taking the minimum over all $t \in \{1,2,\dots,a-1\}$ gives the inequality
    \begin{align*}
      M_{\min}(a) \leq \min_{1 \leq t \leq a-1} M_{\min}(a,t).
    \end{align*}
    Conversely, let $T^* \in \arg\min_{T \in \mathcal{T}_a} S(T)$.
    Let the two subtrees of the root of $T^*$ be $T_\ell^*$ and $T_r^*$,
    and suppose that they have $t^*$ and $a-t^*$ leaves, respectively.
    Since $S(T_\ell^*) \geq M_{\min}(t^*)$ and $S(T_r^*) \geq M_{\min}(a-t^*)$,
    we have
    \begin{align*}
      M_{\min}(a) \stackrel{\eqref{eq:def-m-min-a}}{=} S(T^*)
      &\stackrel{\text{\Cref{lem:sx-recurrence}}}{=}
           F(a,t^*)
           + \binom{a}{t^*} \cdot (a-t^*)! \cdot S(T_\ell^*)
           + \binom{a}{a-t^*} \cdot (t^*)! \cdot S(T_r^*)\\
      &\stackrel{\phantom{\text{\Cref{lem:sx-recurrence}}}}{\geq}
           F(a,t^*)
           + \binom{a}{t^*} \cdot (a-t^*)! \cdot M_{\min}(t^*)
           + \binom{a}{a-t^*} \cdot (t^*)! \cdot M_{\min}(a-t^*)\\
      &\hspace{12pt}\stackrel{\eqref{eq:def-m-min-a-t}}{=}\hspace{12pt}
      M_{\min}(a,t^*) \geq \min_{1 \leq t \leq a-1} M_{\min}(a,t).
    \end{align*}
    This proves the reverse inequality, and hence the equality.
  \end{proof}

  We fill the table $M_{\min}$ in increasing order of $a$.
  For every $a$, we also store a value of $t$ achieving the minimum in
  \eqref{eq:mmin-gs-recurrence}.
  Then a tree achieving $M_{\min}(m)$ can be reconstructed recursively
  from the stored splits.
  By the equivalence at the beginning of the proof,
  this tree maximizes $\outdiv(\GS(T))$.

  It remains to bound the running time.
  The values $F(a,t)$ for all
  $a \leq m$ and $1 \leq t \leq a-1$ are obtained
  using the same preprocessing as in the proof of \Cref{thm:outdiv-gst}:
  we compute the table $A(a,t,s)$ and then apply \Cref{lem:fx-formula}.
  As argued there, this takes $O(m^4)$ time in total.
  Once all values $F(a,t)$ are known,
  filling the table $M_{\min}$ takes $O(m^2)$ time,
  since for every $a$ we try all $t \in \{1,2,\dots,a-1\}$.
  Thus, the overall running time is $O(m^4)$.

  To find a tree minimizing $\outdiv(\GS(T))$,
  we use the same dynamic program with maxima instead of minima.
  Hence we can find both a tree maximizing and a tree minimizing
  $\outdiv(\GS(T))$ in $O(m^4)$ time.
}

Further, using the approach from \Cref{thm:outdiv-gst}, we also obtain
an exact, closed formula for outer diversity of $\GScat$.  Namely, we
apply the algorithm to the caterpillar tree, which yields a recurrence
that we solve explicitly.  On the one hand, using this formula, we can
compute outer diversity of $\GScat$ in $O(m)$ time, much faster than
using \Cref{thm:outdiv-gst}, and, on the other hand, we see that in
the limit it approaches $\nicefrac{1}{2}$.  This extends the result of
\citet{fal-sor-szu-was:t:outer-diversity}, who established that
$\nicefrac{1}{2}$ is a lower bound on the outer diversity of $\GScat$.
To be formally correct, we
use $\GScat_m$ to denote the $\GScat$ domain over $m$ candidates (we
use the same convention for
$\GSbal$, $\SP$, and
$\SPstar$).
\begin{theorem}\label{thm:out-div-gs-cat}
  For every $m \geq 2$, it holds that
  $\textstyle
    \out(\GScat_m)
    = \frac{1}{2} + \frac{1}{m-1}
      - \frac{1}{(m-1)m} \cdot \left( H_m - \frac{1}{2} H_{\lfloor m/2 \rfloor} \right),
  $
  where $H_m = \sum_{i=1}^m 1/i$ is the $m$-th \emph{harmonic number}.
\end{theorem}
\appendixproof{Theorem}{thm:out-div-gs-cat}{
  We fix an integer $m \geq 2$.
  Let $T$ be the caterpillar tree defining $\GScat_m$.
  We compute $\out(\GScat_m)$ by applying the algorithm from
  \Cref{thm:outdiv-gst} to the tree $T$.
  Thus, throughout the proof, we use the notation from the proof of
  \Cref{thm:outdiv-gst}.

  For every $q \in \{2,\dots,m\}$,
  let $x_q$ be the unique node of $T$ whose subtree has exactly $q$ leaves,
  and let $x_1$ be an arbitrary child of $x_2$.
  In particular, $x_m$ is the root of $T$ and $x_1$ is a leaf.
  For every $q \in \{2,\dots,m\}$, the node $x_q$ is internal:
  it has one child that is a leaf and the other child is $x_{q-1}$.
  Since swapping the two children of an internal node does not change the
  induced $\GS$ domain,
  we may assume that the leaf child is the left child.
  Thus, we have $m_{x_q}=q$ and $t=1$.

  To compute the outer diversity,
  we first compute the local contribution $F(x_q)$.
  This value will then be used in the recurrence for $S(x_q)$.
  Since $A(q,1,s)=1$ and $f_{x_q}(s) = \min\{s-1,q-s\}$
  for every $s \in \{1,\dots,q\}$, we get
  \begin{align}
    F(x_q)
    &\stackrel{\text{\Cref{lem:fx-formula}}}{=}
      \sum_{s=1}^q f_{x_q}(s) \cdot A(q,1,s) \cdot 1! \cdot (q-1)!
    = (q-1)! \cdot \sum_{s=1}^q \min\{s-1,q-s\}.\label{eq:fxq-formula}
  \end{align}
  Applying \Cref{lem:sx-recurrence} to $x_q$,
  and using the fact that the leaf child contributes zero, gives
  \begin{align*}
    \frac{S(x_q)}{q!}
    &= \frac{1}{q!} \cdot F(x_q)
       + \frac{1}{q!} \cdot \binom{q}{1} \cdot (q-1)! \cdot 0
       + \frac{1}{q!} \cdot \binom{q}{q-1} \cdot 1! \cdot S(x_{q-1})\\
    &= \frac{F(x_q)}{q!} + \frac{S(x_{q-1})}{(q-1)!}
    \stackrel{\eqref{eq:fxq-formula}}{=}
      \frac{1}{q} \cdot \sum_{s=1}^q \min\{s-1,q-s\}
    + \frac{S(x_{q-1})}{(q-1)!}.
  \end{align*}
  Since $S(x_1)=0$, iterating the above recurrence gives
  \begin{align}
    \frac{S(x_m)}{m!}
    = \sum_{q=2}^m \frac{1}{q}
      \cdot \sum_{s=1}^q \min\{s-1,q-s\}.\label{eq:sxm-over-m-factorial}
  \end{align}
  Now we analyze the term
  $\frac{1}{q} \cdot \sum_{s=1}^q \min\{s-1,q-s\}$.
  If $q$ is even, then we have
  \[
      \frac{1}{q} \cdot \sum_{s=1}^q \min\{s-1,q-s\}
    = \frac{1}{q} \cdot \sum_{s=0}^{\frac{q}{2}-1} 2s
    = \frac{(\frac{q}{2}-1)\frac{q}{2}}{q}
    = \frac{q-2}{4}.
  \]
  If $q$ is odd, then $q=2i+1$ for some integer $i$ and we have
  \[
      \frac{1}{q} \cdot \sum_{s=1}^q \min\{s-1,q-s\}
    = \frac{1}{q} \cdot \left( i + \sum_{s=0}^{i-1} 2s \right)
    = \frac{i^2}{q}
    = \frac{q^2-2q+1}{4q}
    = \frac{q-2}{4}+\frac{1}{4q}.
  \]
  Hence we obtain
  \begin{align}
    \frac{4 \cdot S(x_m)}{m!}
    &\stackrel{\eqref{eq:sxm-over-m-factorial}}{=}
    \sum_{q=2}^m \frac{4}{q} \cdot \sum_{s=1}^q \min\{s-1,q-s\}\nonumber\\
    &=  \sum_{q=2}^{m} (q-2)
      + \sum_{\substack{2 \leq q \leq m\\q \text{ is odd}}} \frac{1}{q}
    = \frac{(m-1)(m-2)}{2}
      +\left(H_m -1 -\frac{1}{2}H_{\lfloor m/2 \rfloor} \right),\label{eq:sxm-over-m-factorial-final}
  \end{align}
  since the second sum is the sum of reciprocals of odd integers between
  $3$ and $m$.
  Substituting this expression into the formula for $\out(\GScat_m)$ gives
  \begin{align*}
    \out(\GScat_m)
    &\stackrel{\eqref{eq:outdiv-gst}}{=}
        1 -\frac{2 \cdot S(x_m)}{\binom{m}{2} \cdot m!}
      = 1 -\frac{1}{(m-1)m} \cdot \frac{4 \cdot S(x_m)}{m!}\\
    &\stackrel{\eqref{eq:sxm-over-m-factorial-final}}{=}
      1-\frac{1}{(m-1)m}
        \cdot \left(
          \frac{(m-1)(m-2)}{2} + H_m -1 -\frac{1}{2}H_{\lfloor m/2 \rfloor}
        \right)\\
    &= \frac{1}{2} + \frac{1}{m-1} -\frac{1}{(m-1)m}
         \cdot \left(
           H_m - \frac{1}{2}H_{\lfloor m/2 \rfloor}
         \right),
  \end{align*}
  which finishes the proof.
}

\begin{corollary}\label{col:cat}
  It holds that
  $\lim_{m \rightarrow \infty}\out(\GScat_m) = \frac{1}{2}$.
\end{corollary}

The case of $\GSbal$ is more complicated and there does not seem to be
an easy way of applying \Cref{thm:outdiv-gst} to obtain an exact
formula for outer diversity.  Nonetheless, a somewhat involved proof
leads to an asymptotic formula and the limit value.

\begin{theorem}
\label{thm:gs-bal-conv}
  For a $\GSbal$ domain over $m$ candidates we have
  $\out(\GSbal) = O(1/\sqrt{m})$.
\end{theorem}
\appendixproof{Theorem}{thm:gs-bal-conv}{
We will prove the thesis by showing that there exists a constant $\beta \in \mathbb{R}$
such that
\begin{equation}
\label{eq:gs-bal-conv:goal}
    \ansd(\GSbal) = 1/2 - \beta/\sqrt{m}.
\end{equation}

For a rooted binary tree $T$ with $C$ as a set of leaves,
let $I(T)$ denote the set of internal nodes with exactly two children in $T$.
Then, for every $x \in I(T)$ by $d_1(x)$ and $d_2(x)$ let us denote
the two children of $x$,
and let $D_i(x) \subseteq C$
be a subset of all leaves
in the subtree rooted in $d_i(x)$ for both $i \in \{1,2\}$
(note that $D_i(x) = \{d_i(x)\}$ if $d_i(x)$ is a leaf).
For clarity, we will assume that $|D_1(x)| \ge |D_2(x)|$.

Let $A$ and $B$  be arbitrary disjoint subsets that form a 2-partition of the set of candidates $C$,
and let us denote $a = |A|$ and $b = |B|$.
By $A,B$-\emph{score} of a vote $v$
let us denote the number of pairs of candidates $c \in A$ and $c' \in B$,
such that $c$ is before $c'$, i.e.,
\[
    \sigma_{A,B}(v) = \sum_{c \in A, c' \in B} [c \succ_v c'],
\]
where $[\varphi]$ is an Iverson bracket giving value $1$ if $\varphi$ is true.
Observe that $\sigma_{A,B}(v)$ is between $0$,
which is obtained when every candidate in $B$ is before every candidate in $A$,
and $a \cdot b$,
which we get when all candidates in $A$ are before that in $B$.
Furthermore, observe that for a vote $v$ taken from $\mathcal{L}(C)$ uniformly at random,
the expected value of its $A,B$-\emph{score} is equal to
\begin{equation}
\label{eq:gs-bal-conv:sigma-exp-value}
    \mathbb{E}(\sigma_{A,B}(v)) = \sum_{c \in A, c' \in B} \mathbb{E}([c \succ_v c']) =\sum_{c \in A, c' \in B} 1/2 = a \cdot b / 2.
\end{equation}
This corresponds to the expected number of swaps needed
to put all candidates in $A$ before that in $B$.
However, we are interested in expected number of swaps
needed to separate candidates $A$ and $B$,
which is what we need to put candidates in $A$ before that in $B$
or vice versa, whichever requires less swaps.
Let us denote this value by
\[
    X(a,b) = \mathbb{E}(\min(\sigma_{A,B}(v), a\cdot b - \sigma_{A,B}(v))).
\]
Note that we use $a$ and $b$ as the arguments for $X$,
as its value depends only on the sizes of $A$ and $B$,
not on the sets themselves.

\begin{lemma}
\label{lem:gs-bal-conv:ansd-as-series}
    For every binary tree $T$ with $m$ leaves it holds that
    \[
        \ansd(\GS (T)) =
        \frac{1}{\binom{m}{2}}\sum_{x \in I(T)} X(|D_1(x)|,|D_2(x)|).
    \]
\end{lemma}
\begin{proof}
    Let $C$ be a set of $m$ candidates, and let $T$ be an arbitrary binary tree over $C$. 
    Consider an arbitrary vote $u \in \mathcal{L}(C)$.
    From the definition of the swap distance we have
    \begin{align*}
        \swap(\GS(T),u) &= 
        \min_{v \in \GS(T)} \swap(v,u) \\ &=
        \min_{v \in \GS(T)} \sum_{\{c,c'\}\subseteq C}
            [(c \succ_v c' \land c' \succ_u c) \lor (c' \succ_v c \land c \succ_u c')].
    \end{align*}
    Now, observe that for every $u,v \in \mathcal{L}(C)$
    it holds that
    \begin{multline*}
        \sum_{\{c,c'\}\subseteq C^2} [(c \succ_v c' \land c' \succ_u c) \lor (c' \succ_v c \land c \succ_u c')] = \\
        \sum_{x \in I(T)} \sum_{c \in D_1(x), c' \in D_2(x)} [(c \succ_v c' \land c' \succ_u c) \lor (c' \succ_v c \land c \succ_u c')],
    \end{multline*}
    as for every pair of candidates $c,c' \in C$
    there is exactly one pair of sets $D_1(x),D_2(x)$
    in which they are counted---the ones in which $x$ is their closest common ancestor.
    Therefore,
    \[
        \swap(\GS(T),u) = 
        \min_{v \in \GS(T)}
        \sum_{x \in I(T)} \sum_{c \in D_1(x), c' \in D_2(x)}
        [(c \succ_v c' \land c' \succ_u c) \lor (c' \succ_v c \land c \succ_u c')].
    \]
    Moreover, observe that for a vote $v$ in $\GS(T)$ domain,
    the decision whether $D_1(x)$ is to be placed before $D_2(x)$
    or after it, is independent from the decision
    how $D_1(y)$ and $D_2(y)$ are to be arranged for any other internal node $y$
    (as they are either pairwise disjoint,
    or $D_1(y),D_2(y) \subseteq D_i(x)$ for $i=1$ or $i=2$).
    Hence, we can exchange the $\min_{v \in \GS(T)}$ and $\sum_{x \in I(T)}$ quantifiers,
    yielding
    \[
        \swap(\GS(T),u) = 
        \sum_{x \in I(T)} 
        \min_{v \in \GS(T)}
        \sum_{c \in D_1(x), c' \in D_2(x)}
        [(c \succ_v c' \land c' \succ_u c) \lor (c' \succ_v c \land c \succ_u c')].
    \]
    Furthermore, since in $v$, for every $x \in I(T)$,
    either every candidate in $D_1(x)$ is before every candidate in $D_2(x)$ or vice-versa,
    this is equivalent to
    \[
        \swap(\GS(T),u) = 
        \sum_{x \in I(T)}
        \min \left(
        \sum_{c \in D_1(x), c' \in D_2(x)}
        [c \succ_u c'],
        \sum_{c \in D_1(x), c' \in D_2(x)}
        [c' \succ_u c]
        \right)
    \]
    Then, the first argument of the minimum is just $(\sigma_{D_1(x),D_2(x)}(u)$,
    by the definition.
    By symmetry, the second argument is equal to
    $|D_1(x)|\cdot |D_2(x)| - \sigma_{D_1(x),D_2(x)}(u)$.
    Therefore, we get
    \[
        \swap(\GS(T),u) = 
        \sum_{x \in I(T)}
        \min(\sigma_{D_1(x),D_2(x)}(u), |D_1(x)|\cdot |D_2(x)| - \sigma_{D_1(x),D_2(x)}(u))
    \]
Taking the expected value over all $u \in \mathcal{L}(C)$
and dividing by $\binom{m}{2}$,
we obtain the thesis.
\end{proof}

\begin{lemma}
\label{lem:gs-bal-conv:partial-convergence}
There exists a constant $\gamma \in \mathbb{R}$,
such that
for every $n \in \mathbb{N}$
and $k \in \{0,\dots,n\}$,
it holds that
    \[\frac{X(n + k, n)}{(n+k)n/2} \ge 1 - \gamma/\sqrt{n}.\]
\end{lemma}
\begin{proof}
Let $v$ be a vote drawn uniformly at random from $\mathcal{L}(C)$
for $C$ consisting of $2n + k$ candidates,
and let $A,B \subseteq C$ be such that
$|A|=n+k$, $|B|=n$, and $A \cap B = \varnothing$.
Then, using $\min(x,y) = \frac{x+y}{2} - \frac{|x-y|}{2}$ transformation,
we obtain that
\begin{align}
    \notag
    X(n + k, n) &= \mathbb{E}(\min(\sigma_{A,B}(v), (n+k)n - \sigma_{A,B}(v))) \\
    \notag
    &= \frac{(n+k)n}{2} -
        \mathbb{E}\left(\left|\sigma_{A,B}(v) -  \frac{(n+k)n}{2}\right|\right).
\end{align}
From Cauchy-Schwarz inequality we get that
\(
    \mathbb{E}(|\sigma(v) - (n+k)n/2|) \le \sqrt{\mathbb{E}((\sigma(x) - ab/2)^2)},
\)
which yields
\[
    X(n + k, n) \ge 
        \frac{(n+k)n}{2} - 
        \sqrt{\mathbb{E}\left(\left(\sigma_{A,B}(v) -  \frac{(n+k)n}{2}\right)^2\right)}.
\]
Since $(n+k)n/2$ is the mean of $\sigma_{A,B}(v)$ by Equation~\eqref{eq:gs-bal-conv:sigma-exp-value},
this gives us
\begin{equation}
\label{eq:gs-bal-conv:exp-value-bound}
    X(n + k, n) \ge \frac{(n+k)n}{2} - \sqrt{\Var\left(\sigma_{A,B}(v)\right)}.
\end{equation}

In order to compute the variance of $\sigma_{A,B}(v)$,
we use the formula for the variance of the sum,
which gives us
\[
    \Var(\sigma_{A,B}(v)) =
    \sum_{c \in A, c' \in B} \sum_{d \in A, d'\in B} \Cov([c \succ_v c'][d \succ_v d']).
\]

Now, if $c = d$ and $c' = d'$, we get
\[
    \Cov([c \succ_v c'][d \succ_v d']) = \Var([c \succ_v c']) = \mathbb{E}(([c \succ_v c'] - 1/2)^2)=(1/2)^2 = 1/4.
\]
If $c \neq d$ and $c' \neq d'$, the variables $[c \succ_v c']$ and $[d \succ_v d']$ are independent,
thus
\[
    \Cov([c \succ_v c'][d \succ_v d']) = 0.
\]
If $c = d$ but $c' \neq d'$, then
\[
    \Cov([c \succ_v c'][d \succ_v d']) = 
    \mathbb{E}\big(([c \succ_v c']-1/2)([c \succ_v d']-1/2)\big).
\]
Observe that $([c \succ_v c']-1/2)([c \succ_v d']-1/2)$
is equal to $1/4$ if $c$ is either before or after both $c'$ and $d'$
while it is equal to $-1/4$ if $c$ is between $c'$ and $d'$.
Hence, out of all $6$ orderings of these $3$ candidates,
$1/4$ is given in $4$ cases and $-1/4$ in $2$.
Thus,
\[
    \Cov([c \succ_v c'][d \succ_v d']) = 
    \frac{1}{4}\cdot\frac{4}{6} - \frac{1}{4}\cdot\frac{2}{6} = \frac{1}{4}\cdot \frac{2}{6} = \frac{1}{12}.
\]
If $c \neq d$ but $c' = d'$, the situation is analogous and
\[
    \Cov([c \succ_v c'][d \succ_v d']) = \frac{1}{12}.
\]

It remains to count the multiplicity of each case with non-zero covariance.
The number of cases when $c = d$ and $c' = d'$
is the number of pairs $c \in A$, $c' \in B$,
which is $(n+k)n$.
Each case in which $c = d$ but $c' \neq d'$
corresponds to picking one candidate from $A$ and two distinct ones from $B$,
but we count $(c',d')$ and $(d',c')$ both times,
which gives $2\binom{n+k}{2}n = (n+k)(n+k-1)n$ cases.
Analogously, there are $2(n+k)\binom{n}{2} = (n+k)n(n-1)$ cases
in which $c \neq d$ but $c' = d'$.
Thus, we get
\begin{align*}
    \Var(\sigma_{A,B}(v)) &= \frac{(n+k)n}{4} + \frac{(n+k)(n+k-1)n + (n+k)n(n-1)}{12}\\
    &= \frac{(3n^2 + 3nk) + (n^3+2n^2k+nk^2 - n^2 - nk) + (n^3 +n^2k -n^2 -nk)}{12} \\
    &= \frac{2n^3 + n^2 + nk +3n^2k+nk^2}{12}
\end{align*}

Combining this with Equation~\eqref{eq:gs-bal-conv:exp-value-bound}, gives us
\begin{align*}
    X(n + k, n) &\ge \frac{(n + k)n}{2} - \sqrt{\frac{2n^3 + n^2 + nk +3n^2k+nk^2}{12}} \\
                &\ge \frac{(n + k)n}{2} - \sqrt{\frac{6n^3 + 2n^2}{12}} \\
                & = \frac{(n + k)n}{2} - \sqrt{\frac{3n^3 + n^2}{6}},
\end{align*}
where the second inequality comes from the fact that $k \le n$.
Dividing it by $(n+k)n/2$ we obtain
\begin{align*}
    \frac{X(n + k, n)}{(n+k)n/2} &\ge 1 -  \sqrt{\frac{3n^3 + n^2}{6}}\cdot \frac{2}{(n+k)n} \\
    &\ge 1 -  \sqrt{\frac{3n^3 + n^2}{6}}\cdot \frac{2}{n^2}\\
    &= 1- O(1/\sqrt{n}),
\end{align*}
where the second inequality comes from the fact that $k \ge 0$.
This concludes the proof of the lemma.
\end{proof}

\begin{lemma}
\label{lem:gs-bal-conv:m-series-convergence}
For every binary tree $T$
with $m$ leaves
it holds that
    \[\sum_{x \in I(T)} |D_1(x)|\cdot |D_2(x)| = \binom{m}{2}.\]
\end{lemma}
\begin{proof}
Let $T$ be an arbitrary tree with a set of leaves $C$, and let $m= |C|$.
Observe that
\[
    \sum_{x \in I(T)} |D_1(x)|\cdot |D_2(x)| = | \{c,c'\} \subseteq C| = \binom{m}{2},
\]
as each pair $c,c' \in C$ is counted exactly once in the sum on the left-hand-side,
for $x$ that is the closest common ancestor of $c$ and $c'$.
This concludes the proof.
\end{proof}

Having \Cref{lem:gs-bal-conv:ansd-as-series,lem:gs-bal-conv:partial-convergence,lem:gs-bal-conv:m-series-convergence},
we are now ready to prove the main thesis.

Fix an arbitrary balanced binary tree $T$ with $m$ leaves.
Observe that in a binary tree, it holds that $|D_1(x)| \le 2 \cdot |D_2(x)|$,
for every internal node $x$.
Thus, we obtain
{
\small
\begin{align*}
    \tag{by \Cref{lem:gs-bal-conv:ansd-as-series}}
    \ansd(\GS (T)) &= \frac{1}{\binom{m}{2}}\sum_{x \in I(T)} X(|D_1(x)|,|D_2(x)|) \\
    &= \frac{1}{\binom{m}{2}}\sum_{x \in I(T)} \frac{|D_1(x)|\cdot|D_2(x)|}{2} \cdot \frac{X(|D_1(x)|,|D_2(x)|)}{|D_1(x)|\cdot|D_2(x)|/2} \\
    \tag{by \Cref{lem:gs-bal-conv:partial-convergence}}
    &\ge \frac{1}{\binom{m}{2}}\sum_{x \in I(T)} \frac{|D_1(x)|\cdot|D_2(x)|}{2} \cdot \left( 1 - \frac{\gamma}{\sqrt{|D_2(x)|}}\right) \\
    \tag{as $|D_1(x)| \le 2 |D_2(x)|$}
    &\ge \frac{1}{\binom{m}{2}} \!\sum_{x \in I(T)}\!
        \left(
            \frac{|D_1(x)|\!\cdot\!|D_2(x)|}{2} \! - \!
            \gamma \cdot \frac{|D_2(x)|^2}{\sqrt{|D_2(x)|}}\!
        \right)\\
    &= \frac{1}{2}\frac{1}{\binom{m}{2}}
    \left(
    \sum_{x \in I(T)}
            |D_1(x)|\cdot|D_2(x)|
    -
    \gamma \cdot 
    \sum_{x \in I(T)}
            |D_2(x)|^{3/2}
        \right)\\
    \tag{by superadditivity}
    &= \frac{1}{2}\frac{1}{\binom{m}{2}}
    \left(
    \sum_{x \in I(T)}
            |D_1(x)|\cdot|D_2(x)|
    - \gamma \left(\sum_{x \in I(T)}|D_2(x)|\right)^{3/2}
    \right)\\
    \tag{as $\sum_{x \in I(T)}|D_2(x)|\le m$}
    &= \frac{1}{2}\frac{1}{\binom{m}{2}}
    \left(
    \sum_{x \in I(T)}
            |D_1(x)|\cdot|D_2(x)|
    \right)
    - \gamma\cdot\frac{m^{3/2}}{m^2 - m}\\
    \tag{by \Cref{lem:gs-bal-conv:m-series-convergence}}
    &= \frac{1}{2}
    - \gamma\left(1 + \frac{m}{m^2 - m}\right)\frac{1}{\sqrt{m}}\\
    \tag{as $m/(m^2-m)\le 1$ for $m\ge 2$}
    &= \frac{1}{2}
    - \frac{2\gamma}{\sqrt{m}}.
\end{align*}
}
This proves \Cref{eq:gs-bal-conv:goal} for $\beta = 2\gamma$ and concludes the proof.
}

\begin{corollary}\label{col:bal}
  It holds that
  $\lim_{m \rightarrow \infty}\out(\GSbal_m) = 0$.
\end{corollary}

In \Cref{fig:swap}, we see that for sufficiently many candidates,
outer diversity of $\SP$ is lower than that of $\GSbal$. This
reinforces our belief that
$\lim_{m \rightarrow \infty}\outdiv(\SP_m)=0$, albeit a proof is still
missing. Another interesting consequence of \Cref{col:cat,col:bal} is
that if $\GScat$ indeed has the highest outer diversity among
Condorcet domains with a given number of candidates, then as we
considered larger and larger candidate sets, Condorcet domains would
achieve outer diversity values between $0$ and $\nicefrac{1}{2}$. In
fact, we would even expect outer diversity values to be placed densely
in between $0$ and $\nicefrac{1}{2}$. To see this, consider, e.g., $\GS$ domains
based on a tree where some fraction of candidates forms a balanced
tree, with one leaf replaced by a caterpillar tree over the remaining
candidates. Such domains would have outer diversity values between $0$
and $\nicefrac{1}{2}$.

We conclude this section by considering $\SPstar$, our representative
of $\SPT$ domains (and, more broadly, a representative of weak
Condorcet domains). This domain is very close to the general
one. Indeed, if $c$ is the candidate placed in the center of the star
graph, then $\SPstar$ includes all votes that rank $c$ either first or
second. Consequently, we have the following results.

\begin{theorem}\label{thm:out-div-sp-star}
  It holds that
  $\out(\SPstar_m) = 1 - 2/m + 4/m^2$.
\end{theorem}
\appendixproof{Theorem}{thm:out-div-sp-star}{
  Let $C$ be a set of $m$ candidates, and let $\SPstar_m$ be the domain over $C$.
  Let $c \in C$ be the candidate corresponding to the center of the star.
  We observe that $\SPstar_m$ is exactly the set of votes in $\calL(C)$
  in which $c$ is ranked first or second.
  Moreover, distance from any vote $v \in \calL(C)$ to $\SPstar_m$ is simply $\max(v[c]-2,0)$,
  where $v[c]$ is the position of candidate $c$ in vote $v$.
  This comes from the fact that it is enough and necessary to move candidate $c$ to the 2nd position (in cases it is not on the 1st or the 2nd position in $v$---then the distance from $v$ to $\SPstar$ is $0$).
  For a given position $i \in [m]$,
  there are $(m-1)!$ many votes in $\calL(C)$ with candidate $c$ on position $i$.
  This gives
  \[\mathbb{E}(\swap(v,\SPstar_m)) = \frac{(m-1)!}{m!} \cdot \sum_{i = 3}^m (i-2) = \frac{(m-1)(m-2)}{2m}.\]
  This implies
  \[\ansd(\SPstar_m) = \frac{\mathbb{E}(\swap(v,\SPstar_m))}{\binom{m}{2}} = (m-2)/m^2,\]
  which gives $\out(\SPstar_m) = 1 - 2/m + 4/m^2$.
}

\begin{corollary}
  It holds that
  $\lim_{m \rightarrow \infty} \out(\SPstar_m) = 1$.
\end{corollary}

While these results are not surprising on their own, they are quite
interesting in the context of the preceding discussion: For weak
Condorcet domains, the range of outer diversity values is from $0$ to
$1$. Using an analogous meta-argument as for the case of $\GS$, but
for star graphs with one leaf replaced by a path, we expect these values
to be placed densely in this interval.

\section{Conclusions}

Our main conclusion is that outer diversity provides valuable insights
into the nature of Condorcet domains. In particular, there seems to be
a strong correlation between the number of never-middle conditions
satisfied by a domain and its outer diversity. On the other hand,
domain size is only weakly correlated with outer diversity.  Our
results strongly suggest that the group-separable/caterpillar domain
has the highest outer diversity value among Condorcet
domains. Somewhat weaker arguments suggest that Black's single-peaked
domain has the highest outer diversity value among Arrow's
single-peaked ones.
Finally, we observed some radical differences between Condorcet and
weak Condorcet domains. For example, there are weak Condorcet domains
whose outer diversity approaches $1$---namely the SP on a star graph
domain---whereas we do not believe that Condorcet domains with such
high values exist (albeit we did not establish this formally).

\section*{Acknowledgments}

Tomasz Wąs was supported by UK Engineering and Physical Sciences Research Council (EPSRC) under grant EP/X038548/1.
Piotr Faliszewski and Krzysztof Sornat were supported by NCN project AGREE (UMO-2025/58/A/ST6/00371).
During its early stages, this project
has received funding from the European Research Council (ERC) under
the European Union’s Horizon 2020 research and innovation programme
(grant agreement No 101002854).
\begin{center}
  \includegraphics[width=3.5cm]{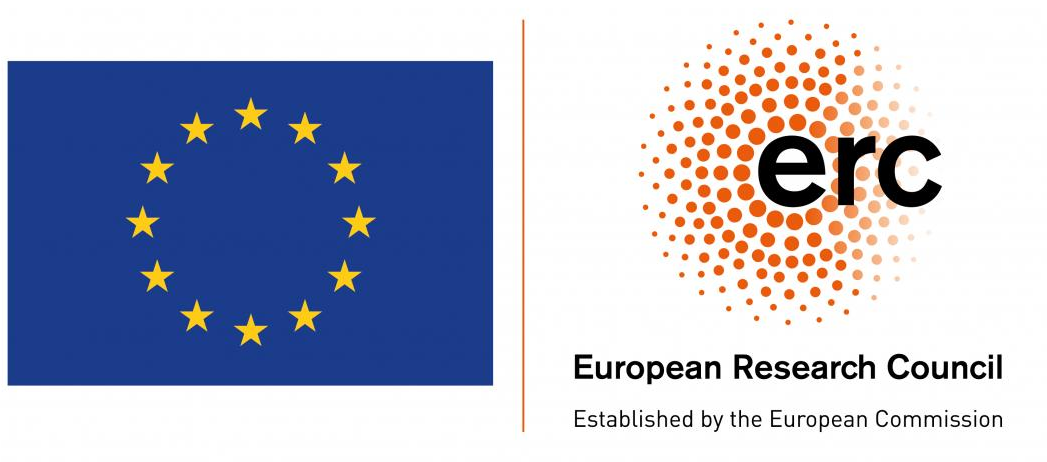}
\end{center}


\appendix
\crefalias{section}{appendix} 

\clearpage

\section{Missing Proofs}
\label{app:missing-proofs}
\appendixProofs

\section{Condorcet Domains with Seven Candidates}
\label{app:seven-candidates}

\begin{figure}[t]
    \centering
    \begin{subfigure}[t]{0.485\textwidth}
        \centering
        \includegraphics[height=5cm]{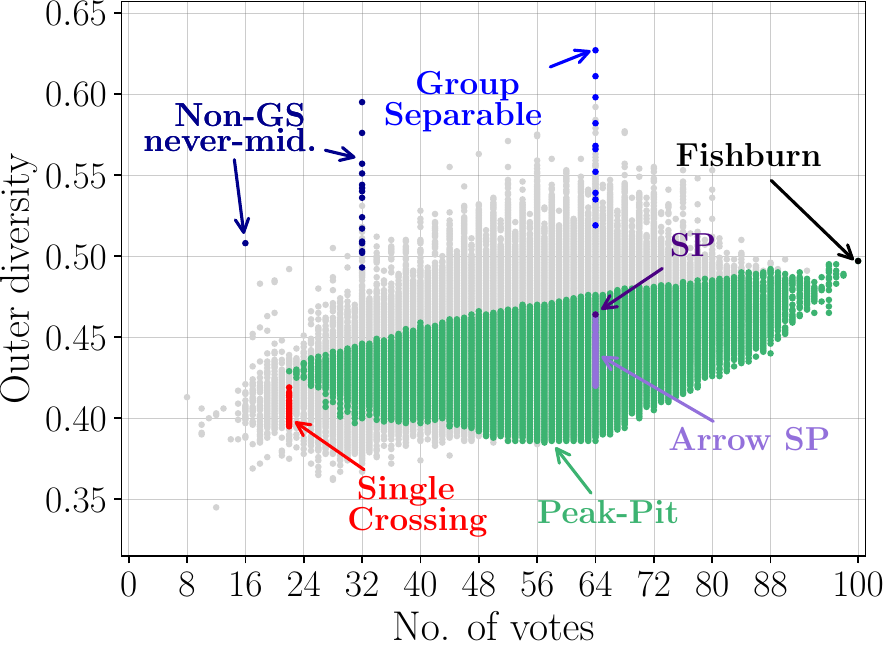}
        \caption{\label{fig:cond7}Every maximal nonisomorphic
          Condorcet domain with seven candidates is placed as a dot, with
          its $x$ coordinate giving the domain's size,
          and its $y$ coordinate giving its outer diversity. If a
          domain belongs to a certain class of domains, such as peak-pit,
          then its dot is colored accordingly; otherwise it is
          gray. We additionally include a sample of 200
          single-crossing domains.}
    \end{subfigure}
    \hfill
    \begin{subfigure}[t]{0.485\textwidth}
        \centering
        \includegraphics[height=5cm]{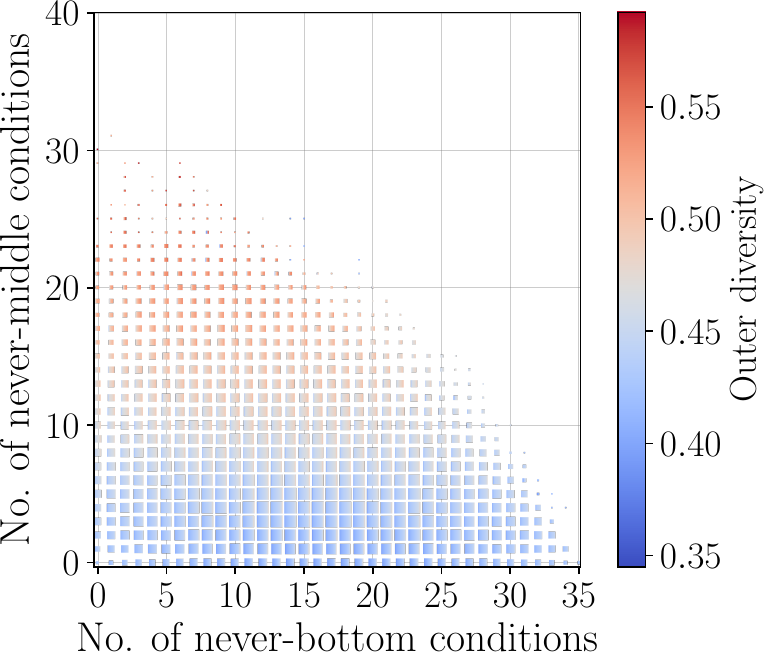}
        \caption{\label{fig:botmid7}Outer diversity of maximal
          nonisomorphic Condorcet domains depending on the number of
          satisfied never-bottom ($x$-axis) and never-middle
          ($y$-axis) conditions. The size of each square reflects the
          number of domains satisfying given numbers of conditions, and
          the horizontal color gradient presents the
          distribution of outer diversity values for these domains.}
    \end{subfigure}
    \caption{Analysis of all of the maximal Condorcet domains for seven candidates.}
    \label{fig:fullcond7}
\end{figure}

In \Cref{fig:fullcond7} we show a plot analogous to
\Cref{fig:fullcond6}, but for the case of seven candidates. Since
there are over 172 million maximal nonisomorphic Condorcet domains
with seven candidates, we included each such domain provided by
\citet{ake-lee-lit-mar-rii:j:condorcet7} with probability
$1\%$. However, in \Cref{fig:cond7} (but not \Cref{fig:botmid7}) we
did include all peak-pit domains, all Arrow's SP domains (including
Black's SP one), all domains that only satisfy never-middle conditions
(including group-separable ones) and the Fishburn domain.
We see that the overall conclusions one can draw for the case of seven
candidates are similar as those for six.

\section{Outer Diversity of Euclidean Domains}
\label{app:euclidean}

Let us formally introduce the Euclidean domains.
\begin{description}
    \item[$\boldsymbol t$D-Euclidean.]
    Let $x \colon C \to \mathbb{R}^t$ be a function that places the candidates as distinct points in $\mathbb{R}^t$. We call a vote $v$ \emph{consistent with $x$} if there is a point $x_v \in \mathbb{R}^t$ from which the candidates appear in the order prescribed by $v$ when sorted by distance. A $t$-dimensional Euclidean domain defined by $x$ contains exactly the votes consistent with $x$. 
\end{description}
Aside from the 1D case, $t$D Euclidean domains are typically neither Condorcet nor weak Condorcet. In the 1D case they are, however, special cases of both $\SP$ and $\SC$ ones.
There are very many different Euclidean domains that depend on exact 
locations of the candidates. Further, we occasionally restrict them
to only include those votes that can be generated by points in some subarea of $\mathbb{R}^t$.

For our experimental analysis, we consider two distinct subareas of 
$\mathbb{R}^t$ to serve as the sampling space for candidate 
locations. The first is the standard square-shaped subarea, 
typically defined as $[-1, 1]^t$. However, this geometry may 
introduce unwanted artifacts due to its lack of rotational symmetry; 
specifically, the increased distance from the center to the vertices 
can bias the resulting domain. For instance, a rotation of the 
candidate set $x$ could shift the regions associated with certain 
consistent votes in or out of the subarea, thereby altering the 
domain. To mitigate these effects, we also consider a circular 
subarea, which ensures rotational invariance.

To investigate how the spatial concentration of candidates 
influences the outer diversity of the induced domain, we utilize a symmetric 
Beta distribution $B(\alpha, \beta)$, where $\alpha=\beta$. By 
varying the parameter~$\alpha$, we can smoothly transition between 
distributions that are concentrated at the center of the subarea 
($\alpha > 1$) and those that are concentrated at the edges ($\alpha < 1$).

The points are generated according to the following procedures:
\begin{description}
\item[Circular Subarea.] To maintain radial consistency, for each 
point we sample an angle $\theta$ uniformly from $[0, 2\pi)$ and a 
radial distance $r$ as the square root of a value sampled from 
$B(\alpha, \alpha)$.

\item[Square Subarea.] The generation method adapts to the
  concentration parameter. For $\alpha > 1$, both coordinate axes are
  sampled independently from $B(\alpha, \alpha)$ and scaled to the
  interval $[-1, 1]$. For $\alpha \le 1$, to prevent
  over-concentration in the corners, we sample one axis from
  $B(\alpha, \alpha)$ and the other from a uniform distribution on
  $[-1, 1]$. To maintain symmetry, this assignment of distributions to
  axes is swapped for exactly half of the generated points.
\end{description}

\begin{figure}[htbp]
    \centering
    \begin{subfigure}[b]{0.48\textwidth}
        \centering
        \includegraphics[width=\textwidth]{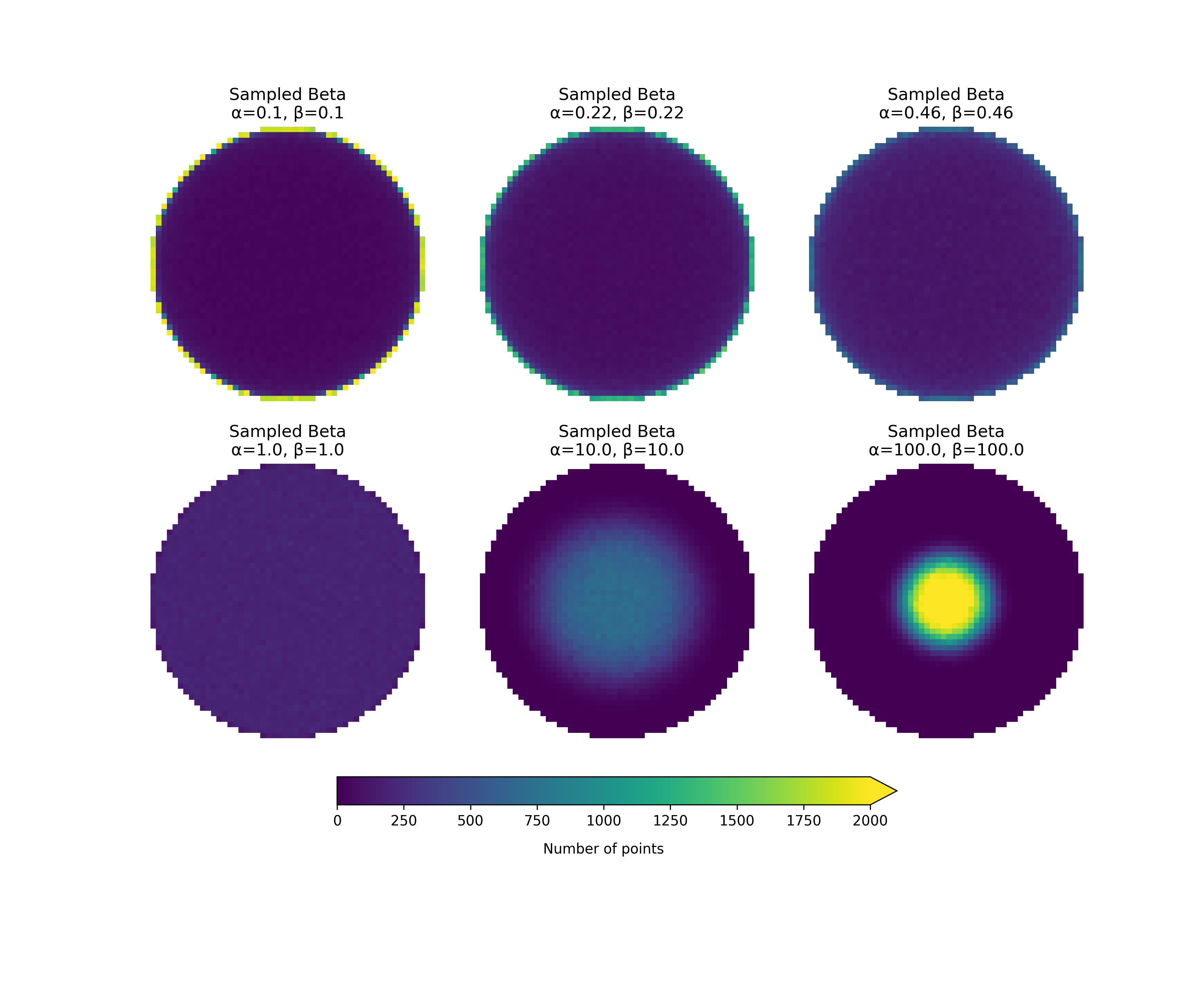}
        \caption{Distribution of points on a circular subarea with chosen parameters.}
        \label{fig:circle_test}
    \end{subfigure}
    \hfill 
    \begin{subfigure}[b]{0.48\textwidth}
        \centering
        \includegraphics[width=\textwidth]{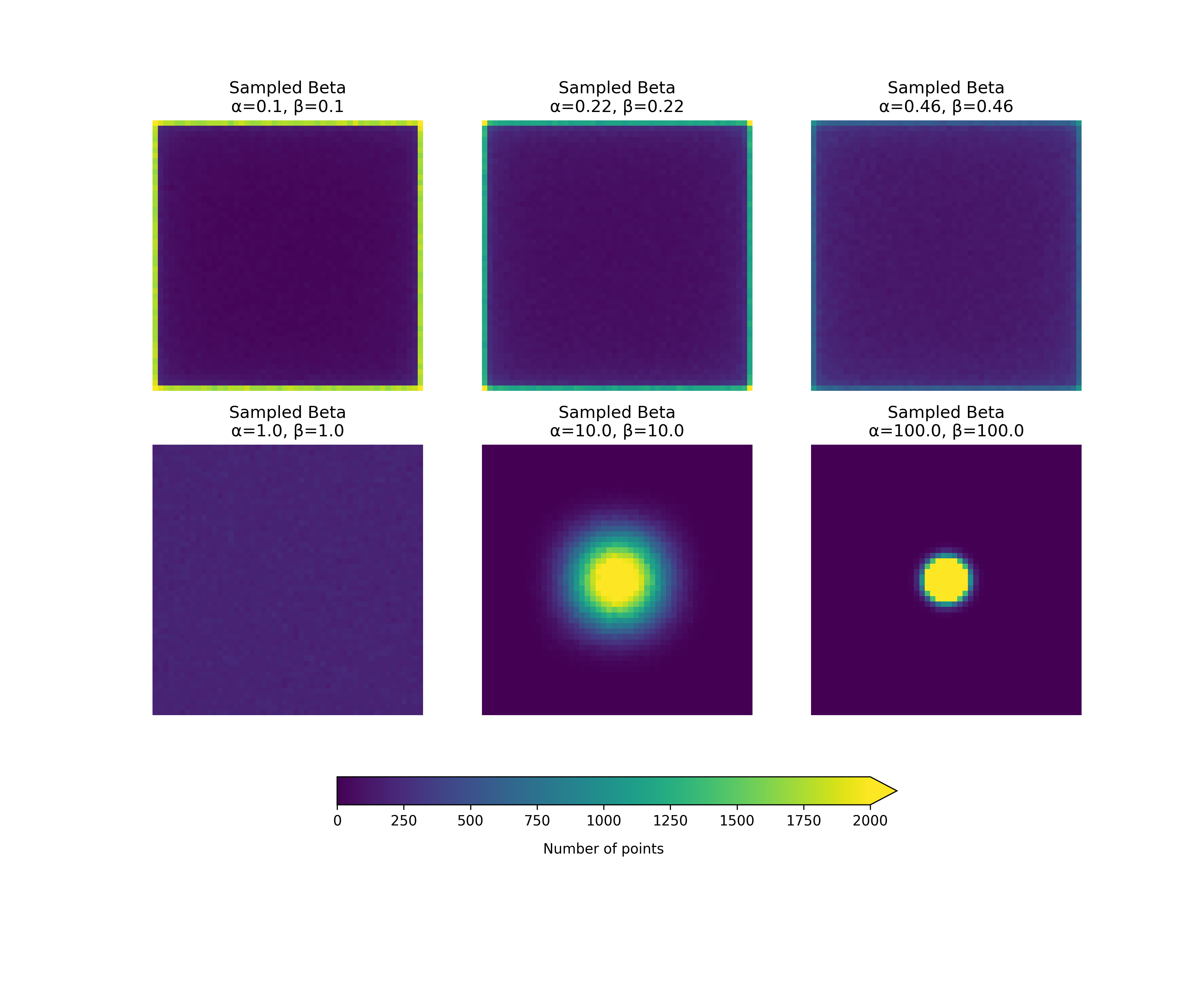}
        \caption{Distribution of points on a square subarea with chosen parameters.}
        \label{fig:square_test}
    \end{subfigure}
    
    \caption{\label{fig:betadistr}Candidate distributions within circular and square subareas using a sample size of $500,000$ points.}
\end{figure}

\begin{figure}[htbp]
    \centering
    \begin{subfigure}[b]{0.48\textwidth}
        \centering
        \includegraphics[width=\textwidth]{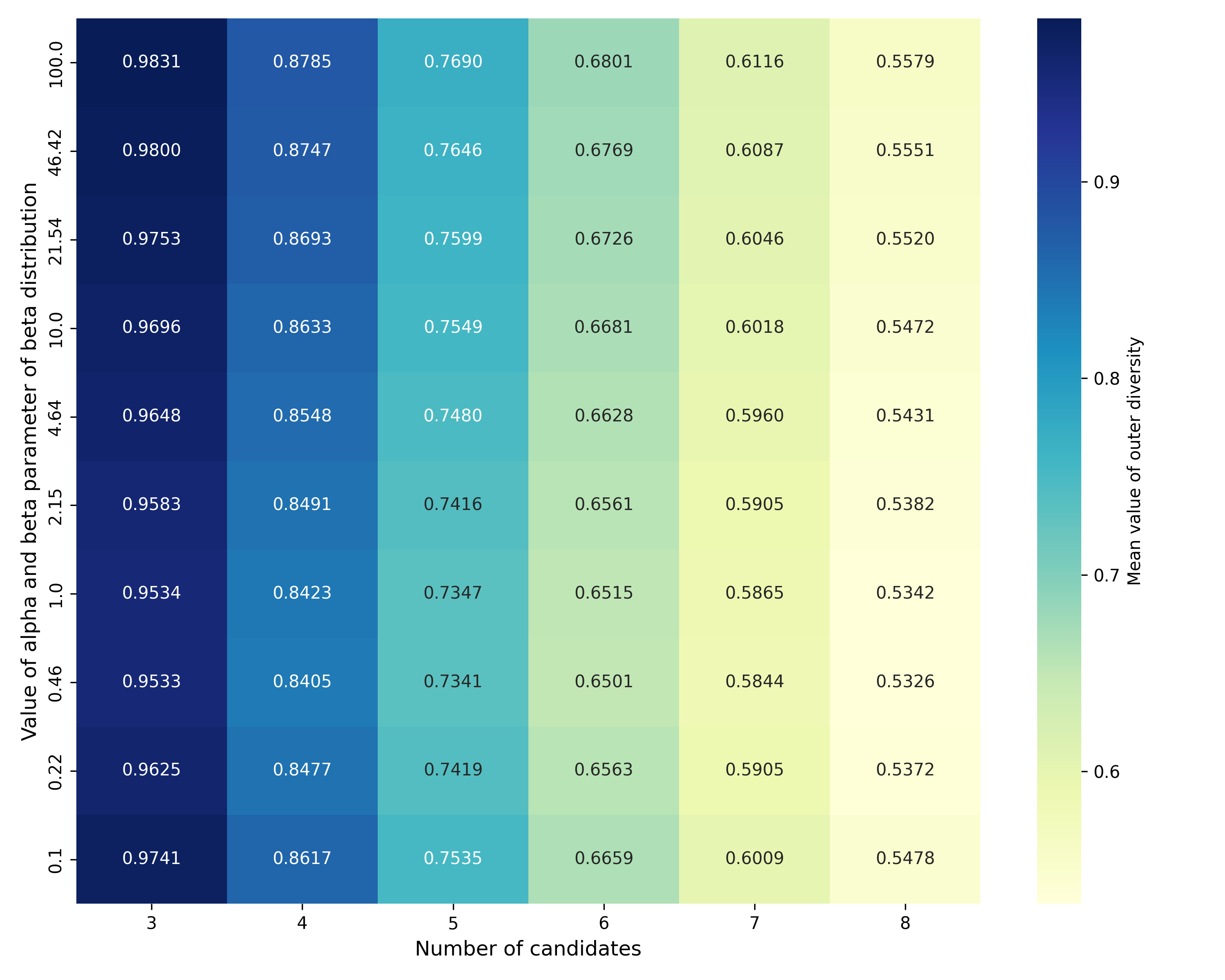}
        \caption{2D-Euclidean circular subarea}
        \label{fig:outer_diversity}
    \end{subfigure}
    \hfill 
    \begin{subfigure}[b]{0.48\textwidth}
        \centering
        \includegraphics[width=\textwidth]{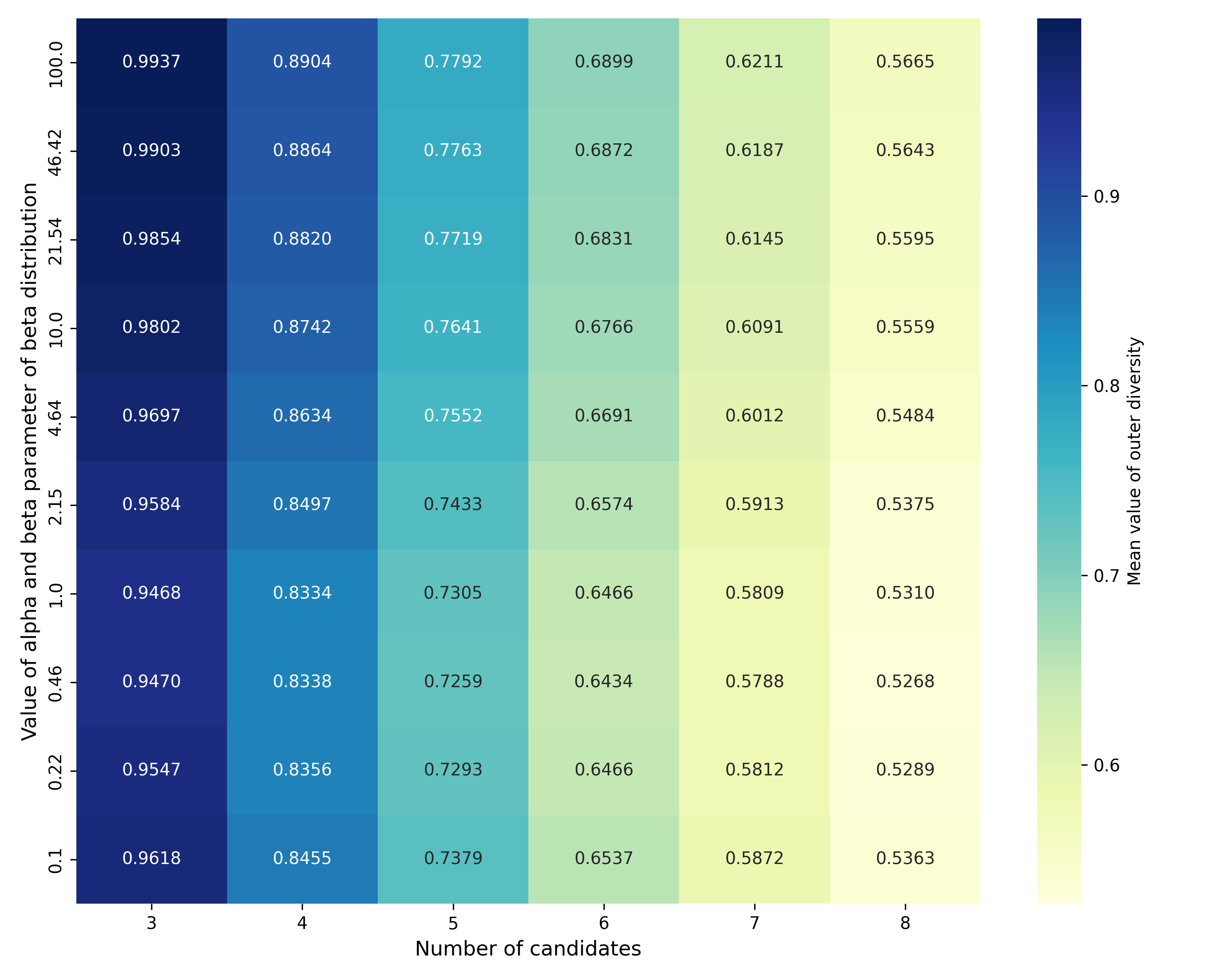}
        \caption{2D-Euclidean square subarea}
        \label{fig:heatmap_circle}
    \end{subfigure}
    
    \caption{\label{fig:betaheat}Heatmaps of mean outer diversity as a function of the number of candidates in 2D Euclidean domains, compared across circular and square subareas. The Beta distribution parameters are shown on a logarithmic scale. The $95\%$ confidence interval indicates a margin of error of at most $0.00262$, with only two exceptions reaching $0.00928$ and $0.00817$ for the square subarea at $\alpha = \beta = 0.1$ with 4 and 3 candidates, respectively. The number of generated samples per candidate count is $6000$, $4000$, $2000$, $1000$, $500$, and $200$ for $3, 4, 5, 6, 7$, and $8$ candidates, respectively.}
    \label{fig:comparison}
\end{figure}

In \Cref{fig:betadistr}, we present the spatial distributions as described above 
for the parameter settings utilized in our subsequent experiments. 
As illustrated, for identical parameter values, the distributions 
exhibit different point concentrations between the circular and 
square subareas — an effect particularly noticeable at $\alpha = \beta = 10.0$. 
However, this discrepancy does not affect our conclusions, 
as we focus on the change in outer diversity relative to the shift
in distribution, rather than absolute values for specific parameter configurations.

We calculated the mean outer diversity for both circular and square subareas, 
generating candidate positions from distributions described above 
for $m \in \{3, \ldots, 8\}$ candidates. The results, presented in \Cref{fig:betaheat}, 
show that regardless of the subarea's geometry, the outer diversity 
is higher when candidates are more concentrated at the center 
compared to the uniform distribution observed at $\alpha = \beta = 
1.0$. 
This result is expected, as a higher concentration of candidates near the center increases the number of geometric regions corresponding to individual votes that lie at least partially within the chosen subarea.

Surprisingly, an analysis of the heatmaps reveals a local minimum in 
the mean outer diversity to be likely for parameters between $0.22$ 
and $1.0$. Contrary to intuitive expectations of a monotonic 
increase as concentration in the center intensifies, the mean outer 
diversity rises again as the parameters decrease further toward 
$0.1$. This non-monotonic behavior is consistent across both the 
circular and the square subareas. In the square subarea for 3 and 4 candidates, the lowest 
mean outer diversity is observed for $\alpha = \beta = 1.0$, yet the 
results for $\alpha = \beta = 0.46$ remain very close, confirming 
that the slight, negligible variations do not 
alter the overall trend.

We also checked the variety of domains that can be obtained. \Cref{fig:eddiversityvotes} presents the variety of 2D Euclidean domains in terms of outer diversity and the number of votes, obtained by generating candidates uniformly at random.

\begin{figure}[H]
    \centering
    \includegraphics[width=0.9\textwidth]{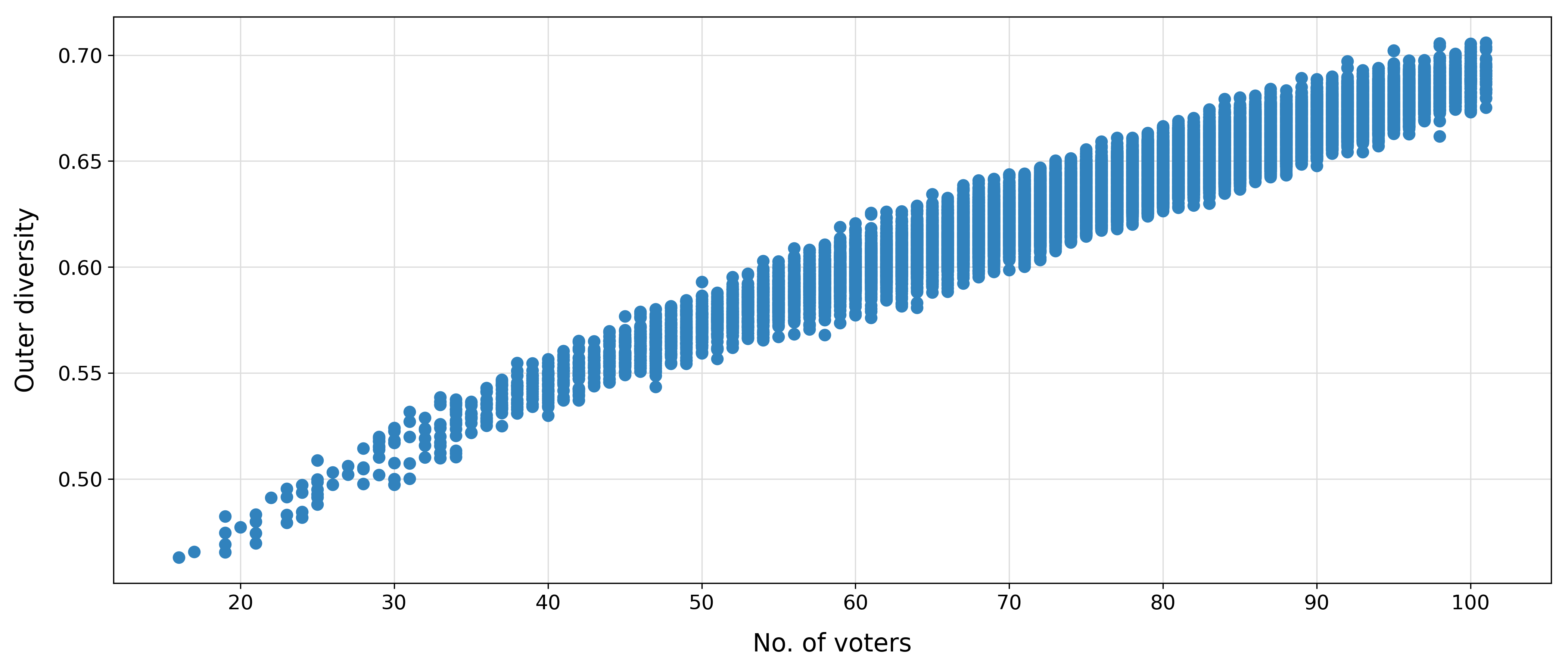}
    \caption{\label{fig:eddiversityvotes}$100,000$ 2D Euclidean domains with candidates generated uniformly at random within a square subarea. Each domain is represented as a single dot, where the x-axis corresponds to the number of votes within the domain and the y-axis represents its outer diversity. This subset illustrates a representative sample, not all possible domains.}
\end{figure}

\newpage
\section{Microscope Visualizations of the Domains}
\label{app:microscope}

\Cref{fig:microscope-all6max,fig:microscope-all6min} show microscope visualizations of Condorcet domains on 6 candidates. For each domain size, we display the domains with maximum (Fig.~\ref{fig:microscope-all6max}) and minimum (Fig.~\ref{fig:microscope-all6min}) outer-diversity. Each domain is represented as a point cloud where votes are embedded via MDS under swap distance, with point color indicating the top-ranked candidate. Domains with high outer-diversity appear geometrically spread out with colors well-separated across the embedding, while low-diversity domains of the same size are more clustered, reflecting that their votes occupy a narrower portion of the space. (Note that for $n=45$, there is only one domain, which has both the minimal and maximal outer-diversity).

\Cref{fig:microscope-nevermid-16v,fig:microscope-nevermid-32v} shows microscope domains satisfying never-middle conditions. Note that the domains presented in Fig.~\ref{fig:microscope-nevermid-32v} are at the same time $GS$ domains, and consist of 32 votes, each; while domains presented in Fig.~\ref{fig:microscope-nevermid-32v} are not $GS$ and consist only of 16 votes, each.

Finally, figure~\ref{fig:microscope-spt} presents analogous microscope visualizations for SP/T domains with up to 9 candidates and up to 1000 votes. 

\begin{figure}[h]
    \centering
    \includegraphics[width=\textwidth]{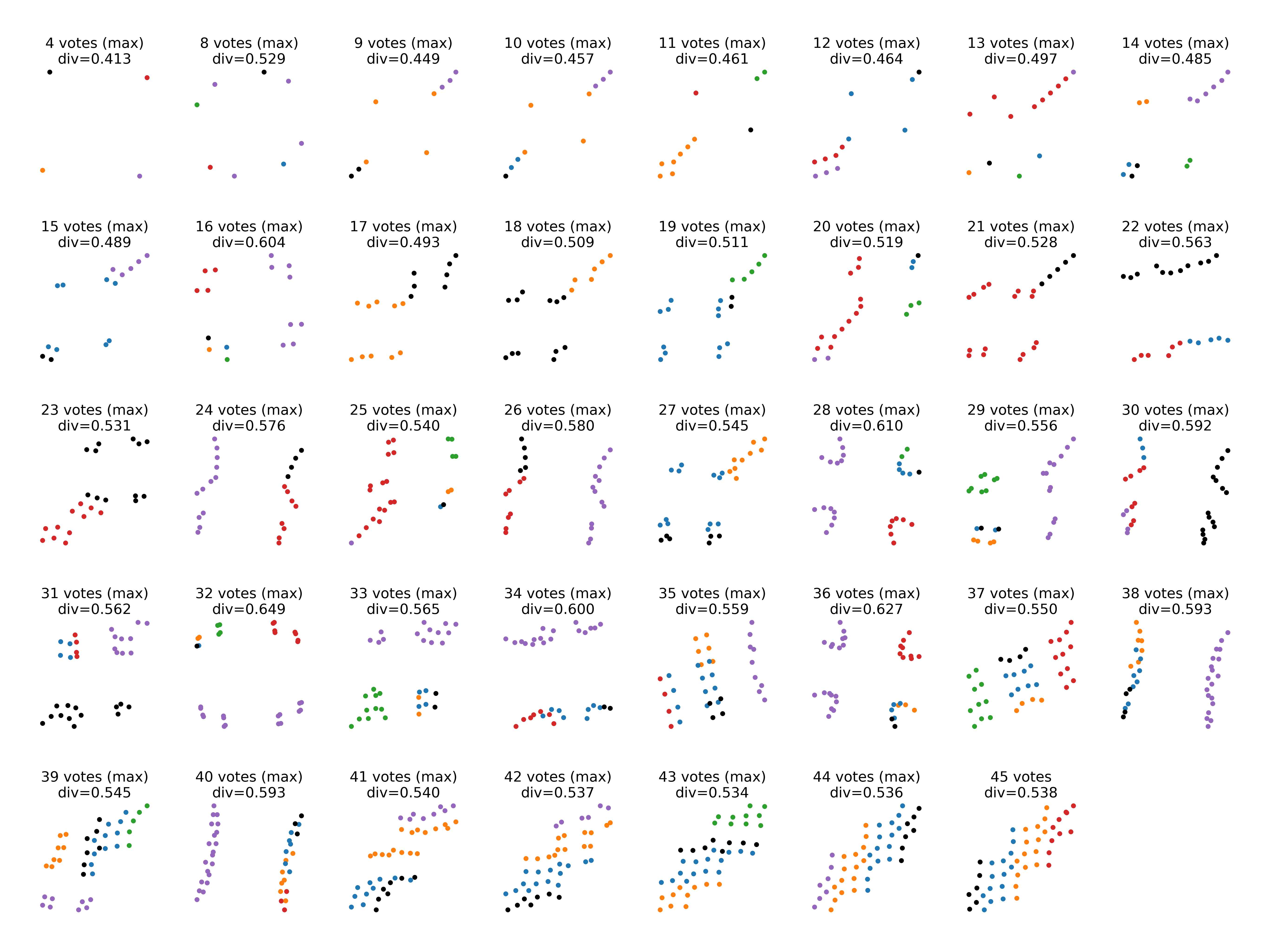}
    \caption{\label{fig:microscope-all6max} Microscope visualizations of Condorcet domains on 6 candidates. For each domain size (number of distinct votes), we show the domain with the largest outer-diversity. Points represent distinct votes embedded via MDS under swap distance; color encodes the top-ranked candidate.}
\end{figure}

\begin{figure}[ht]
    \centering
    \includegraphics[width=\textwidth]{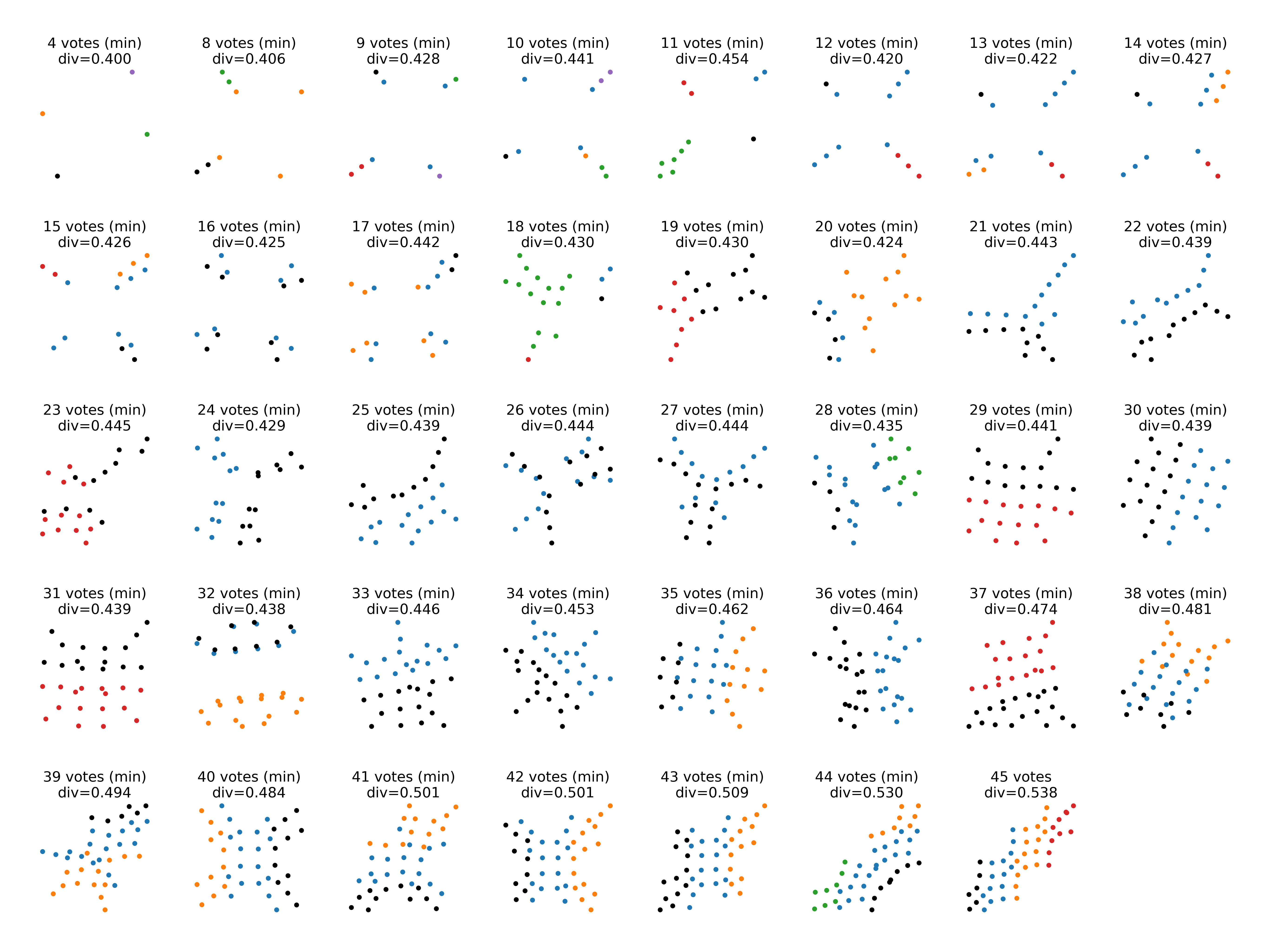}
    \caption{\label{fig:microscope-all6min} Microscope visualizations of Condorcet domains on 6 candidates. For each domain size (number of distinct votes), we show the domain with the smallest outer-diversity. Points represent distinct votes embedded via MDS under swap distance; color encodes the top-ranked candidate.}
\end{figure}

\begin{figure}[t]
    \centering
    \includegraphics[width=0.92\textwidth]{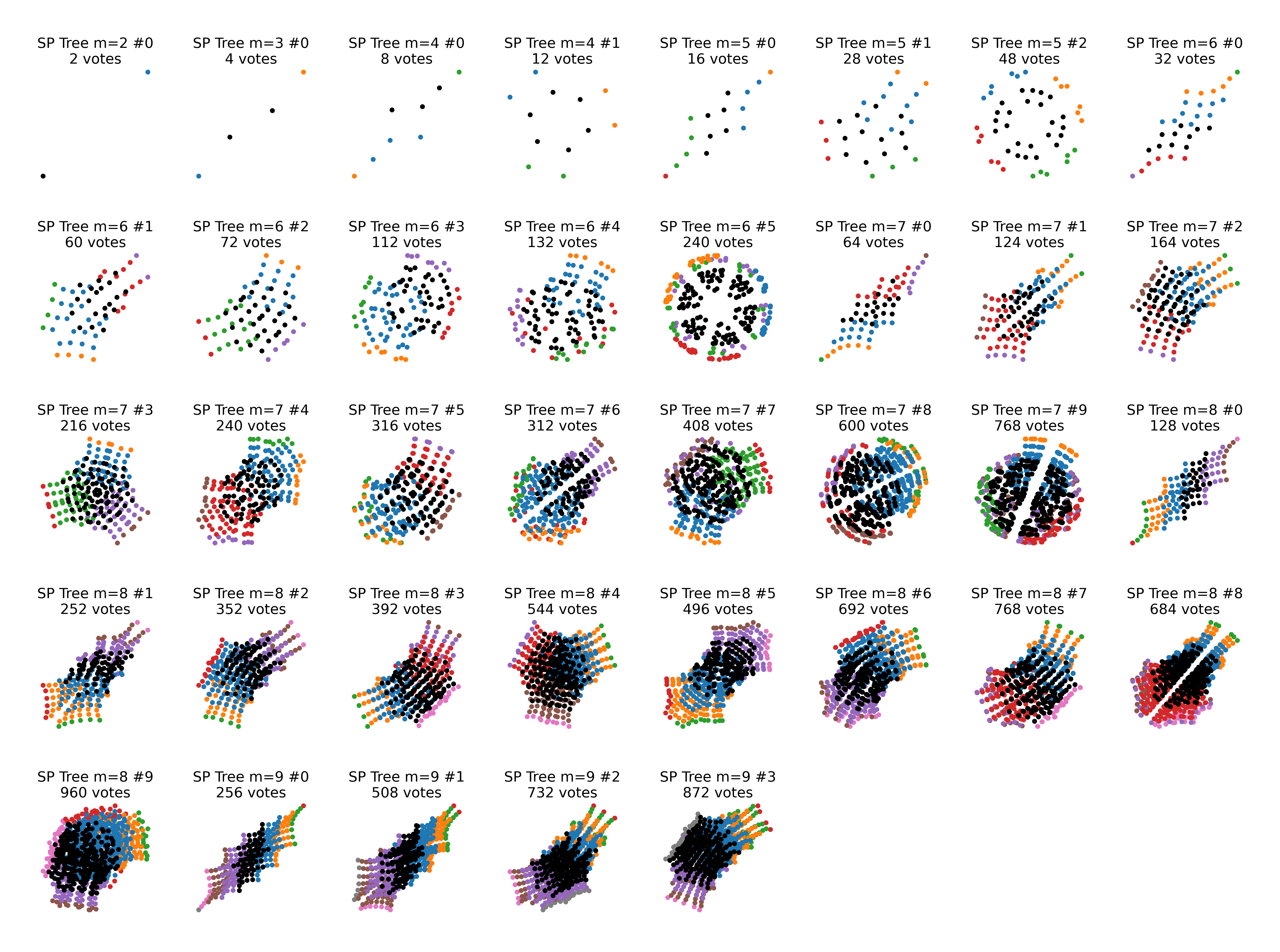}
    \caption{\label{fig:microscope-spt}Microscope visualizations of various SP/T domains with up to 9 candidates, and up to 1000 votes. For each domain size (number of distinct votes), we show the domain with maximum (top) and minimum (bottom) outer-diversity. Points represent distinct votes embedded via MDS under swap distance; color encodes the top-ranked candidate.}
\end{figure}

\begin{figure}[H]
    \centering
    \includegraphics[width=\textwidth]{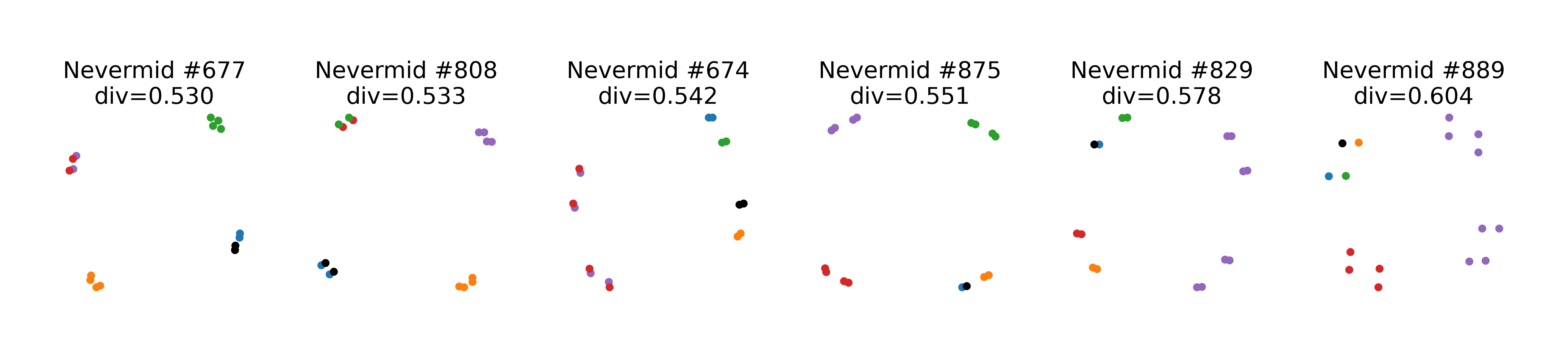}
    \caption{\label{fig:microscope-nevermid-16v}Microscope visualizations of domains satisfying never-middle condition, that are not $GS$. Points represent distinct votes embedded via MDS under swap distance; color encodes the top-ranked candidate.}
\end{figure}

\begin{figure}[H]
    \centering
    \includegraphics[width=\textwidth]{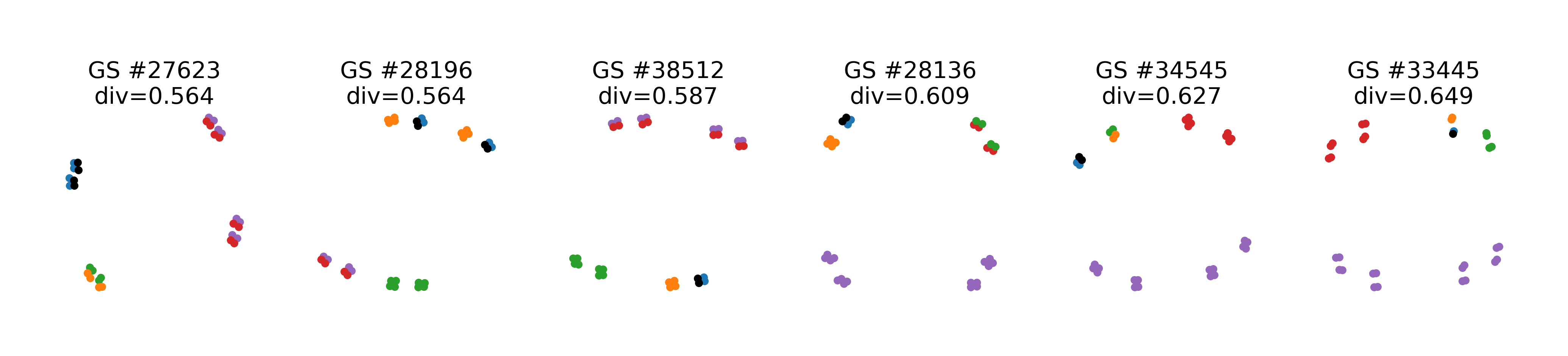}
    \caption{\label{fig:microscope-nevermid-32v}Microscope visualizations of $\GS$ domains. Points represent distinct votes embedded via MDS under swap distance; color encodes the top-ranked candidate.}
\end{figure}

\end{document}